\documentclass[pdflatex,sn-mathphys-num]{sn-jnl}% Math and \usepackage{graphicx}%
\usepackage{comment}
\usepackage{multirow,soul}%
\usepackage{amsmath,amssymb,amsfonts,a4wide}%
\usepackage{amsthm}%
\usepackage{mathrsfs}%
\usepackage[title]{appendix}%
\usepackage{xcolor,lineno}%
\usepackage{textcomp}%
\usepackage{manyfoot}%
\usepackage{booktabs}%
\usepackage{algorithm}%
\usepackage{algorithmicx}%
\usepackage{algpseudocode}%
\usepackage{listings}%
\theoremstyle{thmstyleone}%
\newtheorem{theorem}{Theorem}%  meant for continuous numbers
\newtheorem{proposition}[theorem]{Proposition}% 

\theoremstyle{thmstyletwo}%
\newtheorem{remark}{Remark}%
\newtheorem{lemma}{Lemma}%

\theoremstyle{thmstylethree}%
\newcommand{\Omegamu}{\Omega}
\begin{document}

\title[Resolving Algebraic Complexity]{An Analytically Tractable Framework for Multi-Strain Epidemics: Resolving Algebraic Complexity to Map Oscillatory Dynamics}

\author*[1]{\fnm{Nir} \sur{Gavish}}\email{ngavish@technion.ac.il}
\affil*[1]{\orgdiv{Faculty of Mathematics}, \orgname{Technion Israel Institute of Technology}, \orgaddress{\city{Technion City, Haifa}, \postcode{3200003}, \country{Israel}}}

 \abstract{Multi-strain epidemiological systems frequently exhibit self-sustained oscillations, yet severe algebraic complexity has long obstructed a complete analytical characterization of these dynamics. Seeking to bypass these barriers, we have identified a broad, analytically tractable class of two-strain models featuring asymmetric cross-immunity that precludes secondary infections for a single strain. This targeted structural simplification allows us to derive explicit expressions for coexistence equilibria and map their stability. We show that the relative transmission advantage of secondary infections fundamentally governs the onset of robust limit-cycle oscillations, extending previous narrow-borderline results. Crucially, numerical simulations reveal a rich macroscopic landscape where small-amplitude local oscillations coexist with large-amplitude, recurrent outbreak cycles. The structural robustness of these phenomena is confirmed by incorporating effects such as waning immunity and isolation and extending the framework beyond the analytically tractable class. These findings provide a clear mechanistic explanation for multi-strain epidemic cycles and offer a highly tractable baseline for future theoretical developments.}

\maketitle

\section{Introduction}\label{sec:intro}
Many infectious diseases, such as seasonal influenza and dengue fever, involve multiple co-circulating strains that interact dynamically through host immunity~\cite{martcheva2015introduction,wormser2008modeling,aguiar2022mathematical,Keeling-book}. While the principle of competitive exclusion predicts that less fit strains should be driven to extinction, biological reality frequently defies this prediction. Instead, multi-strain systems frequently exhibit persistent coexistence and complex, self-sustained oscillations.  Understanding the biological mechanisms that generate and sustain such oscillatory behavior is a fundamental issue in epidemiology, with important implications for forecasting epidemic trajectories and designing sustainable public health strategies.

A key candidate for driving these self-sustained oscillations is partial cross-immunity, an immunological interaction where prior infection with one strain confers incomplete protection against others. Mathematical models have shown that cross-immunity can generate a wide range of dynamical behaviors. Notably, models with additional structural complexity, such as age stratification, quarantine, or temporary cross-immunity tailored to specific pathogens, readily exhibit complex, self-sustained oscillations ~\cite{gumel2008existence,garba2013cross,hussaini2016mathematical,ferguson1999effect,gupta1998chaos,nuno2005dynamics,nuno2008mathematical,thieme2007pathogen,kuddus2022analysis,castillo1989epidemiological,kooi2014analysis,kooi2023multi,kooi2013bifurcation,aguiar2011role,aguiar2013how,aguiar2008epidemiology}. In contrast, the foundational two-strain SIR framework with lifelong partial cross-immunity was initially found to produce only damped oscillations~\cite{andreasen1997dynamics, castillo1989epidemiological, nuno2005dynamics}, and later work showed that sustained oscillations in this minimal model occur only under highly restrictive conditions~\cite{chung2016dynamics}.

Recently, Gavish~\cite{gavish2024newoscillatoryregimetwostrain} identified a novel oscillatory regime within this minimal framework by demonstrating the failure of standard asymptotic expansions along a critical curve in parameter space. The modeling assumptions of that study, however, restricted the emergence of these oscillations to a case where one strain confers no cross-immunity to the other, while confining the oscillatory behavior to an extremely narrow parametric region. This specificity and narrowness raises two fundamental questions: are these newly discovered oscillations merely a mathematical artifact of a highly specific parameter set, or do they indicate a broader, biologically robust phenomenon? More generally, as noted in~\cite{chung2016dynamics}, the emergence of self-oscillations is ``rather intriguing, as there seems to be no biological reasoning behind this phenomenon''. The biological rationale driving this behavior has remained elusive, leaving open a related question posed by~\cite{gavish2024newoscillatoryregimetwostrain}: what is biologically special about the newly identified curve of instability?

In this work, we address these questions and elucidate the biological drivers of self-sustained oscillations by systematically investigating a broad family of two-strain transmission models to map their oscillatory dynamics. However, a significant initial obstacle arises: in general models, the coexistence equilibrium is defined only implicitly as the solution to a complex nonlinear algebraic system, which makes extracting stability conditions analytically inaccessible. To overcome this barrier, we identify an analytically tractable subfamily of models that are still rich enough to sustain oscillatory phenomena. Specifically, we achieve this by deliberately precluding secondary infections for one strain.  Although we employ this condition primarily as a theoretical tool to bypass algebraic complexity, such highly asymmetric cross-immunity dynamics map naturally to real-world interactions, such as the competitive exclusion observed between \textit{Bordetella pertussis} and \textit{Bordetella parapertussis}~\cite{wolfe2007antigen}. Establishing this structural frontier provides a clear analytical lens through which to map and study the system's oscillatory behavior.

Utilizing this framework, we demonstrate that the onset of limit-cycle oscillations is governed by the overall relative transmission advantage of a secondary infection. We show how oscillations can emerge even when the two strains are completely symmetric, possessing identical epidemiological features. This finding provides critical context to the results in~\cite{gavish2024newoscillatoryregimetwostrain}. In retrospect, that study evaluated a borderline scenario residing precisely at the onset of instability, explaining why the previously identified oscillatory regime was so parametrically narrow. By relaxing these restrictive historical assumptions and moving away from this mathematical borderline, we uncover parameter regions that support far more robust oscillations.

From an epidemiological perspective, focusing on a family of models that precludes secondary infections for one strain may appear restrictive, given that bidirectional secondary infections are common in many multi-strain systems. However, this specific framework allows us to explicitly incorporate complex biological and behavioral features, such as Antibody-Dependent Enhancement (ADE), waning immunity, and isolation, without rendering the mathematical stability analysis intractable. To demonstrate the utility of this approach, we analyze an extended system incorporating both waning immunity and isolation. We observe that, despite a significantly faster timescale for susceptible replenishment driven by waning immunity, the system preserves the same characteristic oscillatory patterns as the base model. Finally, to test the potential breadth of these findings without exceeding the scope of the current work, we conduct an initial numerical exploration of systems where secondary infections are permitted for both strains, as is typical for influenza~\cite{andreasen1997dynamics}. We observe that the oscillatory dynamics extend to these broader cases, strongly suggesting that the fundamental instability mechanism identified here is structurally robust.

While our study investigates oscillatory dynamics predominantly through the lens of the coexistence equilibrium and its local stability, conforming with standard approaches in mathematical epidemiology, we identified a second, distinct oscillatory mechanism in multiple case studies. This unexpected phenomenon is characterized by large-amplitude recurrent outbreaks. We provide clear numerical observations of these macroscopic cycles and demonstrate that they can coexist with small-amplitude local oscillations. Characterizing the boundaries of these global transitions remains a compelling open problem. Although a full investigation of these global dynamics falls outside the scope of the present work, we use the discussion section to outline the separate global mechanism driving these macroscopic cycles, thereby establishing a foundational scaffold for a companion paper dedicated to characterizing this global phenomenon.

This paper is organized as follows. In Section~\ref{sec:considerations}, we introduce the criteria for selecting our analytically tractable family of two-strain models and present the specific base model. Section~\ref{sec:steady_states} provides explicit expressions for the single-strain and coexistence equilibria. In Section~\ref{sec:stability_phiCE_base}, we investigate the stability of the coexistence steady state employing a perturbation analysis to identify the parameter regimes that lead to oscillatory instabilities. We complement this analysis with a numerical study in Section~\ref{sec:numeric_base_model} to explore the different types of oscillatory solutions that emerge. The robustness of these findings is evaluated in Section~\ref{sec:extensions_numerical}, which considers an extended model incorporating waning immunity and quarantine measures, and in Section~\ref{sec:general_two_strain}, which considers models beyond the analytically tractable family of models. Finally, Section~\ref{sec:discussion} provides a summary and discussion of our implications, highlighting the global oscillatory mechanisms to be explored in a companion paper.

\subsection{Two-Strain Transmission Base Model}\label{sec:gnrl_twostrain_base_model}
We consider, as a basis, the standard two-strain transmission model for the dynamics of partial cross-immunity, where infection by one strain may provide reduced susceptibility to infection by the other~\cite{castillo1989epidemiological, martcheva2015introduction, thieme2007pathogen}. Significantly, we extend this base model by allowing the recovery time from secondary infections to differ from the recovery time from primary infections.  

\begin{figure}[ht!]
\begin{center}
\includegraphics[width=0.75\textwidth]{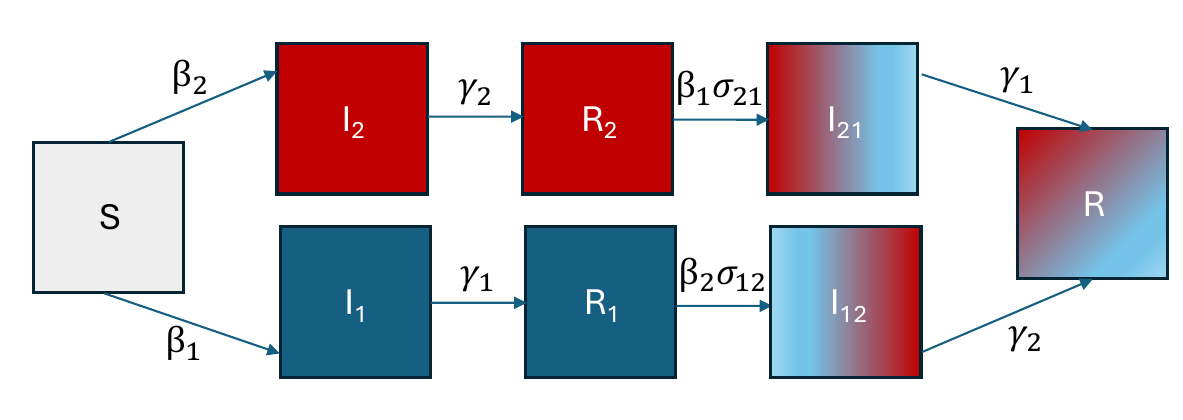} 
\caption{Schematic diagram of disease dynamics when the host is exposed to two co-circulating strains. Susceptibles ($S$) may be infected with strain~$i=1$ or~$i=2$ ($I_i$, primary infection). Those recovered from strain~$i$ ($R_i$) are immune to reinfection by strain~$i$ but may be susceptible to infection by strain~$j\ne i$ ($I_{ji}$, secondary infection). Here,~$\sigma_{ij}$ is the relative susceptibility to strain~$j$ for an individual previously infected with and recovered from strain~$i$ ($i\ne j$). A value of~$\sigma_{ij}>0$ corresponds to partial cross-immunity,~$\sigma_{ij}=1$ corresponds to no cross-immunity, and~$\sigma_{ij}>1$ corresponds to enhanced susceptibility. Additionally,~$\gamma_i$ is the recovery rate from strain~$i$, while~$\gamma_{ji}$ is the recovery rate from a secondary infection with strain~$i$.}
\label{fig:diagram_both}
\end{center}
\end{figure} 

As shown in Figure~\ref{fig:diagram_both}, we assume the population mixes randomly and is divided into eight relative compartments: Susceptibles ($S$); individuals infected with a primary infection of strain~$i$ ($I_i$); individuals recovered from a primary infection with strain~$i$ ($R_i$); individuals infected with strain~$i$ after having recovered from strain~$j\ne i$ ($I_{ji}$, secondary infection); and individuals recovered from both strains ($R$). The model is given by the following system of equations:
\begin{subequations} \label{eq:model}
\begin{align}
        \frac{\text{d}S}{\text{d}t} &=\mu(1-S)-\sum_{i=1}^{2}\beta_i(I_i+\eta_i I_{ji})S,\label{eq:eqS}\\[1ex]
         \frac{\text{d}I_1}{\text{d}t}  &= \beta_1S(I_1+\eta_1I_{21})-(\mu+\gamma_1)I_1,\label{eq:eqI1}\\[1ex]
         \frac{\text{d}I_2}{\text{d}t}  &= \beta_2S(I_2+\eta_2I_{12})-(\mu+\gamma_2)I_2,\label{eq:eqI2}\\[1ex]
         \frac{\text{d}I_{21}}{\text{d}t}  &= \beta_1\sigma_{21}R_2(I_1+\eta_1 I_{21})-(\mu+\gamma_{21})I_{21},\label{eq:eqJ1}\\[1ex]
         \frac{\text{d}I_{12}}{\text{d}t}  &= \beta_2\sigma_{12}R_1(I_2+\eta_2I_{12})-(\mu+\gamma_{12})I_{12},\label{eq:eqJ2}\\[1ex]
         \frac{\text{d}{R_1}}{\text{d}t} &= \gamma_1I_1-\beta_2\sigma_{12}(I_2+\eta_2I_{12})R_1-\mu R_1,\label{eq:eqR1}\\[1ex]
         \frac{\text{d}{R_2}}{\text{d}t} &= \gamma_2I_2-\beta_1\sigma_{21}(I_1+\eta_1I_{21})R_2-\mu R_2,\label{eq:eqR2}\\[1ex]
         \frac{\text{d}{R}}{\text{d}t} &=\gamma_{21}I_{21}+\gamma_{12}I_{12}-\mu R,\nonumber
\end{align}
\end{subequations}
where~$\mu$ is the birth and mortality rate,~$\beta_i$ is the transmission coefficient for strain~$i$,~$\gamma_i$ is the recovery rate from strain~$i$, and~$\gamma_{ji}$ is the recovery rate from a secondary infection with strain~$i$. Additionally,~$\sigma_{ij}$ is the relative susceptibility to strain~$j$ for an individual previously infected and recovered from strain~$i$ ($i\ne j$), where~$\sigma_{ij}>0$ corresponds to reduced susceptibility (partial cross-immunity),~$\sigma_{ij}=1$ corresponds to neutral susceptibility (no cross-immunity) and~$\sigma_{ij}>1$ corresponds to enhanced susceptibility. The relative infectivity of individuals with a secondary infection of strain~$i$, compared to those with a primary infection, is denoted by~$\eta_i$.

\section{Model Selection: Criteria and Identification}\label{sec:considerations}
We wish to identify a family of analytically tractable models  suitable for studying oscillatory phenomena.  We begin by listing the primary criteria for determining such a family of two-strain models:
\begin{enumerate}
    \item {\bf Asymmetry Between Strains:} Symmetric models are commonly studied because they are mathematically tractable~\cite{chung2016dynamics, castillo1989epidemiological}, yet they may not accurately represent real-world epidemiological dynamics where strains exhibit distinct biological profiles~\cite{aguiar2011role,kooi2014analysis}. Even introducing weak asymmetry can reveal novel phenomena that purely symmetric models overlook, as demonstrated by the narrow oscillatory regime recently identified in~\cite{gavish2024newoscillatoryregimetwostrain}. Therefore, we seek asymmetric models with the belief that they are likely to provide a clearer, more robust framework to capture and analyze the generic oscillatory phenomena driven by these biological imbalances.
    \item {\bf Explicit Expression for the Coexistence Steady-State:} The endemic steady-state of coexistence is the solution to a nonlinear algebraic system. While numerical studies have long shown limit cycles in asymmetric models, the algebraic complexity of this equilibrium has historically precluded the derivation of explicit, analytical criteria for the onset of these oscillations. Implicit or cumbersome expressions for the steady-state variables significantly complicate stability analysis or make it intractable~\cite{chung2016dynamics,castillo1989epidemiological,martcheva2015introduction,gavish2024newoscillatoryregimetwostrain}. Finding a family of models characterized by an explicit and simple coexistence steady-state is a technical necessity to unlock the algebraic thresholds governing these oscillations.
    \item {\bf Broad Family of Models:} To gain insight from multiple models, we require a broad family of models that satisfies the above two requirements, rather than a single carefully selected model.
\end{enumerate}

To meet these objectives, a delicate mathematical balance must be struck.  While satisfying all three criteria simultaneously poses a challenge, we demonstrate that preventing secondary infections with one strain provides a straightforward mathematical scalpel to achieve this. Let us consider model~\eqref{eq:model} with~$\sigma_{21}=0$. Requirement 1 is met because the model is inherently asymmetric, allowing secondary infections only with strain 2. Requirement 2 is also satisfied, as the coexistence steady-state becomes defined by a clear and straightforward expression. By definition, $I_{21}\equiv0$ when secondary infections with strain 1 are impossible. As a result, Equation~\eqref{eq:eqI1} simplifies to
\[
 \frac{\text{d}I_1}{\text{d}t}  = \left[\beta_1S-(\mu+\gamma_1)\right]I_1,
\]
leading to a simple expression for $S$ at equilibrium:
\[
S=\frac{\mu+\gamma_1}{\beta_1}=\frac{1}{\mathcal{R}_1}.
\]
Substitution of the above value for~$S$ removes nonlinear terms such as $S I_i$ from the algebraic equations describing the coexistence steady-state, allowing one to explicitly solve the remaining linear system, as detailed in Section~\ref{sec:steady_states}. Finally, the assumption that eliminates secondary infections for one strain can be applied across a broad family of two-strain models. For example, in Section~\ref{sec:extensions_numerical}, we find explicit expressions for the coexistence steady-state in an extended two-strain model with quarantine and waning immunity. Thus, requirement 3 is also satisfied.

Although mathematical necessity narrowed our focus from the general two-strain case to this specific model family, we emphasize that this choice is primarily an analytical tool. It is well established that general two-strain models can exhibit complex dynamics, including Hopf bifurcations and chaos, even without external forcing \cite{aguiar2008epidemiology, aguiar2011role, kooi2013bifurcation}. Therefore, precluding secondary infections for one strain ($\sigma_{21}=0$) is not introduced to claim structural novelty, but rather to explicitly bypass the severe algebraic complexity that typically obscures the coexistence equilibrium. Nevertheless, this structural simplification remains firmly rooted in biological reality. Strongly asymmetric cross-immunity dynamics have been documented in real-world pathogen interactions, such as the competitive exclusion observed between \textit{Bordetella pertussis} and \textit{Bordetella parapertussis} \cite{wolfe2007antigen}. Ultimately, this overarching methodological strategy, deriving a specific, analytically tractable model family from objective criteria to illuminate macroscopic behaviors, proves highly effective here. As subsequent sections will demonstrate, this targeted approach successfully extracts the robust analytical mechanisms that drive epidemic cycles.

\subsection{Model of Study}\label{sec:mathModel}

Having established the utility of a model family where secondary infections with one strain are excluded, we now define the specific model for our study. For brevity, we derive this by substituting~$\sigma_{21}=0$ into the general model~\eqref{eq:model} and aggregating the compartments~$R$ and $R_2$ into a single compartment~$R$.

\begin{figure}[ht!]
\begin{center}
\includegraphics[width=0.75\textwidth]{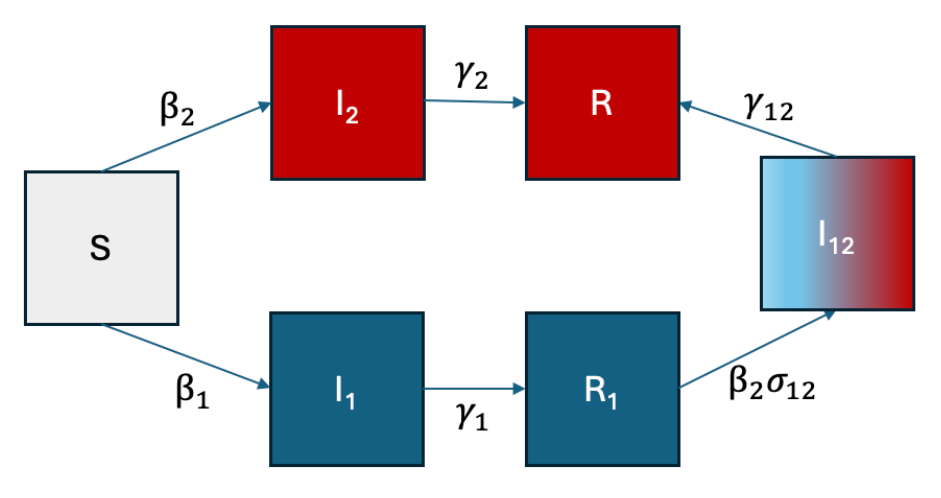}
\caption{Schematic diagram of disease dynamics for the base model. Susceptibles ($S$) may be infected with strain 1 or 2 ($I_i$, primary infection). Those recovered from strain 1 ($R_1$) are immune to reinfection by strain 1 but are susceptible to infection by strain 2. Those recovered from strain 2 ($R$) are immune to reinfection by both strains.}
\label{fig:diagram_reduced}
\end{center}
\end{figure}

The disease dynamics for this reduced model are schematically depicted in Figure~\ref{fig:diagram_reduced} and governed by the following system:
\begin{subequations}\label{eq:base_model}
\begin{align}
        \frac{\text{d}S}{\text{d}t} &=\mu(1-S)-\beta_1 I_1S-\beta_2(I_2+\eta_2I_{12})S,\label{eq:reduced_eqS}\\[1ex]
         \frac{\text{d}I_1}{\text{d}t}  &= \beta_1I_1S-(\mu+\gamma_1)I_1,\label{eq:reduced_eqI1}\\[1ex]
         \frac{\text{d}I_2}{\text{d}t}  &= \beta_2(I_2+\eta_2I_{12})S-(\mu+\gamma_2)I_2,\label{eq:reduced_eqI2}\\[1ex]
         \frac{\text{d}I_{12}}{\text{d}t}  &= \sigma_{12}\beta_2(I_2+\eta_2I_{12})R_1-(\mu+\gamma_{12})I_{12},\label{eq:reduced_eqI12}\\[1ex]        
         \frac{\text{d}{R_1}}{\text{d}t} &= \gamma_1I_1-\sigma_{12}\beta_2(I_2+\eta_2I_{12})R_1-\mu R_1,\label{eq:reduced_eqR1}\\[1ex]
         \frac{\text{d}{R}}{\text{d}t} &=\gamma_2I_2+\gamma_{12}I_{12}-\mu R.
\end{align}
\end{subequations}

We consider initial conditions that satisfy
\begin{equation}\label{eq:reduced_IC}
\begin{split}
&S(0)\ge0,\quad I_i(0)\ge0, \quad I_{12}(0)\ge0, \quad R_1(0)\ge0, \quad R(0)\ge0,\quad i=1,2,\\
&S(0)+I_1(0)+I_2(0)+I_{12}(0)+R_1(0)+R(0)=1,
\end{split}
\end{equation}
which ensures that for all~$t>0$,
\begin{equation}\label{eq:boundOnVariables}
\begin{split}
&S(t)\ge0,\quad I_i(t)\ge0, \quad I_{12}(t)\ge0, \quad R_1(t)\ge0, \quad R(t)\ge0,\quad i=1,2,\\
&S(t)+I_1(t)+I_2(t)+I_{12}(t)+R_1(t)+R(t)\equiv1.
\end{split}
\end{equation}
In what follows, we use the algebraic relation~\eqref{eq:boundOnVariables} to express~$R(t)$ in terms of the other compartment sizes and consider the model as a system of the remaining equations~\eqref{eq:reduced_eqS}--\eqref{eq:reduced_eqR1}.
We consider the model~\eqref{eq:base_model} with vital dynamics~$\mu>0$, restricting the feasible parameter space to strictly positive values:
\begin{equation}\label{eq:parameterspace_reduced_only}
\mu>0,\quad \beta_i>0,\quad \gamma_i>0, \quad \eta_2>0,\quad \sigma_{12}>0,\qquad i=1,2.
\end{equation}

A standard computation of the basic reproduction number, e.g., as in~\cite{chung2016dynamics,castillo1989epidemiological}, yields
\[
\mathcal{R}_0=\max\{\mathcal{R}_1, \mathcal{R}_2\},
\]
where
\begin{equation}\label{eq:RRi}
\mathcal{R}_i=\frac{\beta_i}{\gamma_i+\mu}.
\end{equation}
\section{Steady States}\label{sec:steady_states}
The system~\eqref{eq:base_model} has a disease-free equilibrium,~$\phi^{\rm DFE}=(1,0,0,0,0)$, that is stable, by definition, when~$\mathcal{R}_0<1$, and two single-strain endemic states~$\phi^{EE,i}$ for~$i=1,2$: 
\begin{enumerate}
\item If~$\mathcal{R}_1>1$, there exists a steady-state~$\phi^{EE,1}$ of~\eqref{eq:base_model} with compartment sizes 
\begin{equation}\label{eq:phiEE1}
    S^{EE,1}=\frac1{\mathcal{R}_1},\quad I_1^{EE,1}=\frac{\mathcal{R}_1-1}{(\gamma_1+\mu)\mathcal{R}_1}\mu,\quad R_1^{EE,1}=\frac{\mathcal{R}_1-1}{(\gamma_1+\mu)\mathcal{R}_1}\gamma_1,\quad I_2^{EE,1}=I_{12}^{EE,1}=0.
\end{equation}

\item If~$\mathcal{R}_2>1$, then there exists a steady-state~$\phi^{EE,2}$ of~\eqref{eq:base_model} with compartment sizes
\begin{equation}\label{eq:phiEE2}
    S^{EE,2}=\frac1{\mathcal{R}_2},\quad I_2^{EE,2}=\frac{\mathcal{R}_2 - 1}{(\gamma_2+\mu)\mathcal{R}_2}\mu,\quad I_1^{EE,2}=I_{12}^{EE,2}=R_1^{EE,2}=0.
\end{equation}
\end{enumerate}
The invasion numbers,~$\hat{\mathcal{R}}_j^i$, of strain~$j$ at the endemic equilibrium,~$\phi^{EE,i}$, of strain~$i$ (for~$i=1,2$ and~$j\ne i$) are computed as in~\cite[Ch. 8.4]{martcheva2015introduction}:
\begin{equation}\label{eq:Rinvasion1}
\hat{\mathcal{R}}_2^1=\frac{\mathcal{R}_2}{\mathcal{R}_1}\left[1+\frac{\gamma_1(\gamma_2+\mu)\eta_2\sigma_{12}}{(\gamma_1+\mu)(\gamma_{12}+\mu)}(\mathcal{R}_1-1)\right],\qquad \hat{\mathcal{R}}_1^2=\frac{\mathcal{R}_1}{\mathcal{R}_2}.
\end{equation}

In accordance with standard epidemiological theory, the endemic steady state~$\phi^{EE,i}$ is linearly stable if and only if the invasion reproduction number~$\hat{\mathcal{R}}_i^j$  
of the competing strain~$j\ne i$ is strictly less than one.
\begin{proposition}[Stability of Single-Strain Steady States]\label{prop:single_strain_stability}
    Let~$\mathcal{R}_i>1$ such that the single-strain endemic steady state,~$\phi^{EE,i}$, of system~\eqref{eq:base_model} exists. Then,~$\phi^{EE,i}$ is linearly stable if and only if the invasion number,~$\hat{\mathcal{R}}_j^i$, of the competing strain~$j\ne i$ is strictly less than one:
    \[
    \hat{\mathcal{R}}_j^i<1.
    \]
\end{proposition}
\begin{proof}
See Appendix~\ref{app:proof_phiEE_extended}.  Note that this proof applies to the more general Proposition~\ref{prop:stability_EE_extended}, and reduces to the proof of Proposition~\ref{prop:single_strain_stability} when 
$w=w_1=w_2=\delta_1=\delta_2=0$.
\end{proof}

We, therefore, obtain a partition of the~$(\mathcal{R}_1,\mathcal{R}_2)$ plane into several domains based on the stability of the disease-free and single-strain equilibria, as illustrated in Figure~\ref{fig:Omega_mu}:
\begin{equation}\label{eq:Omega_c}
\begin{split}
\Omega_{c,0}&=\left\{\mathcal{R}_1\le 1,\quad \mathcal{R}_2\le1\right\},\\
\Omega_{c,1}&=\left\{\mathcal{R}_1>1,\quad \hat{\mathcal{R}}_2^1<1\right\},\qquad 
\Omega_{c,2}=\left\{\mathcal{R}_2>1,\quad \mathcal{R}_2>\mathcal{R}_1\right\},
\end{split}
\end{equation}
where~$\Omega_{c,0}$ (gray region) is the domain in which the disease-free equilibrium is stable, and~$\Omega_{c,1}$ (blue region) and~$\Omega_{c,2}$ (red region) are the domains where the single-strain endemic states~$\phi^{EE,1}$ and~$\phi^{EE,2}$ are stable, respectively.   \begin{figure}[ht!]
\begin{center}
\includegraphics[width=0.5\textwidth]{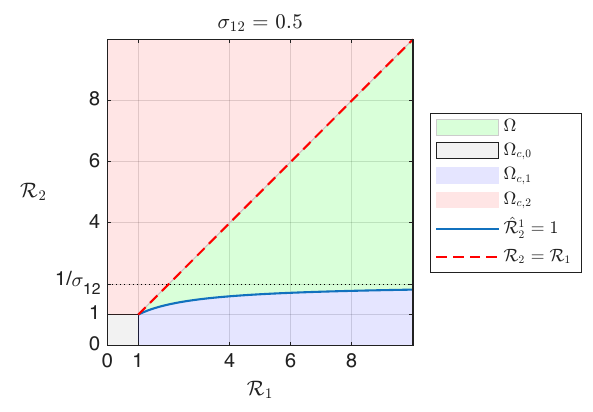}
\caption{
The stability regions~$\Omega_{c,0}$ (gray), $\Omega_{c,1}$ (blue), $\Omega_{c,2}$ (red), defined by~\eqref{eq:Omega_c}, and the coexistence region~$\Omegamu$ (green) defined by~\eqref{eq:omega_mu}. The parameters are chosen such that~$\gamma_2=\gamma_{12}$ and $\sigma_{12}=0.5$ at the limit~$\mu\to0^+$. Note that all other model parameters do not affect the boundaries of these domains in this asymptotic limit.}
\label{fig:Omega_mu}
\end{center}
\end{figure}The following proposition shows the existence of a unique steady-state of coexistence in the complementary regime~$\Omegamu$ (green region), defined as:
 \begin{equation}\label{eq:omega_mu}
 \Omegamu=\left\{\hat{\mathcal{R}}_1^2>1,\quad \hat{\mathcal{R}}_2^1>1\right\}.
 \end{equation}

\begin{proposition}[Coexistence Steady State]\label{prop:phiCE}
    Let~$(\mathcal{R}_1,\mathcal{R}_2)\in \Omegamu$. Then, system~\eqref{eq:base_model} possesses a unique coexistence steady state,~$\phi^{CE}$, characterized by strictly positive infectious populations~$I_1>0$ and~$I_2>0$. The compartment sizes of~$\phi^{CE}$ are given by:
\begin{subequations}\label{eq:phiCE}
\begin{equation}
    \begin{split}
S^{CE}&=\frac1{\mathcal{R}_1},\quad I_1^{CE}=\frac{[1 + \sigma_{12}(\mathcal{R}_1 - 1)]\mu R_1^{CE}}{\gamma_1+(\gamma_1+\mu)\sigma_{12}\mathcal{R}_1 R_1^{CE}},\quad I_2^{CE}=\frac{\mu\,c}{(\gamma_2+\mu)\sigma_{12}}(\hat{\mathcal{R}}^1_2-1)
,\\
I_{12}^{CE}&=\frac{\mathcal{R}_1-\mathcal{R}_2}{\eta_2\mathcal{R}_2}I_2^{CE},\quad  
R_1^{CE} = \frac{(\gamma_{12}+\mu)(\mathcal{R}_1 - \mathcal{R}_2)}{(\gamma_2+\mu)\eta_2\sigma_{12}\mathcal{R}_1\mathcal{R}_2},
\end{split}
\end{equation}
where
\begin{equation}
c=
\frac{(\gamma_1+\mu)(\gamma_{12}+\mu)}{(\gamma_1+\mu)(\gamma_{12}+\mu)(\mathcal{R}_1-\mathcal{R}_2)+\gamma_1\eta_2(\gamma_2+\mu)\mathcal{R}_2}.
\end{equation}  
\end{subequations}
\end{proposition}

\begin{proof}
From~\eqref{eq:reduced_eqI1},~$I_1\ne0$ implies
\[
S^{CE}=\frac1{\mathcal{R}_1}.
\]
From~\eqref{eq:reduced_eqI2},
\[
I_{12}^{CE}=\frac{\mathcal{R}_1 - \mathcal{R}_2}{\eta_2\mathcal{R}_2}I_2^{CE}.
\]
Similarly, the expression for~$R_1^{CE}$ follows directly from~\eqref{eq:reduced_eqI12}.
The remaining system for~$I_1^{CE}$ and~$I_2^{CE}$ is linear, and its solution yields~\eqref{eq:phiCE}.
\end{proof}

\section{Stability of the Coexistence Steady State}\label{sec:stability_phiCE_base}

Direct stability analysis of~$\phi^{CE}$ via the Routh-Hurwitz criteria is intractable due to the complexity of the quintic characteristic polynomial. To address this, we follow established methodologies in the literature~\cite{andreasen1997dynamics, dawes2002onset, chung2016dynamics, gavish2024newoscillatoryregimetwostrain} and analyze the system in the asymptotic regime
\begin{equation}\label{eq:asymptotic_regime_stability}
\mu\ll \gamma_i,\quad \mu\ll \gamma_{12},\qquad i=1,2,
\end{equation}
which corresponds to a demographic turnover rate significantly slower than the disease recovery rates. This separation of time scales is characteristic of acute infections like influenza or dengue, where the infectious period (days/weeks) is negligible compared to the host lifespan (decades)~\cite{CDCyellowbook2024}. Mathematically, this represents a singular perturbation problem, as the coexistence equilibrium vanishes in the limit $\mu \to 0$.

Building on this framework, we apply the singular perturbation approach utilized in~\cite{gavish2024newoscillatoryregimetwostrain}. In that study, the method was applied to a model with mutual partial cross-immunity ($0<\sigma_{ij}<1$) under the restricted assumptions of symmetric parameters for secondary infections: identical recovery rates ($\gamma_{ij}=\gamma_i$) and equal infectivity ($\eta_i=1$). Consequently, that analysis pointed to instability in a rather limited region, specifically when either~$\sigma_{12}\to 1$ or~$\sigma_{21}\to 1$. Here, we apply the same mathematical approach to the strongly asymmetric model ($\sigma_{21}=0$), but we allow for general positive secondary epidemiological parameters ($\gamma_{12}$ and~$\eta_2$). Significantly, this generalization yields vastly different outcomes, uncovering a much broader instability regime that reduces previous work to a borderline case.

The characteristic polynomial of the Jacobian of~\eqref{eq:base_model} at~$\phi^{CE}$ takes the form
\begin{equation}\label{eq:Plambda}
P(\lambda)=\lambda^5 + \frac{\gamma_2(\mathcal{R}_1 - \mathcal{R}_2)+ \gamma_{12}\mathcal{R}_2}{\mathcal{R}_1}\lambda^4+\mu\lambda^2P_1(\lambda)+\mu^2P_2(\lambda)+\mu^3P_3(\lambda)+\mathcal{O}(\mu^4),
\end{equation}
where
\begin{equation}
P_i(\lambda)=P_{i0}+\lambda P_{i1}+\cdots+\lambda^{d_i} P_{id_i},\quad P_{i0}\ne0,\quad i=0,1,2,\cdots,%\sum_{k=0}^5 \lambda^kP_{ik}
\end{equation}
where~$d_i$ is the degree of each polynomial~$P_i(\lambda)$.
The coefficients,~$P_{ij}=P_{ij}(\mathcal{R}_1,\mathcal{R}_2; \gamma_1, \gamma_2, \gamma_{12}, \sigma_{12}, \eta_2)$, depend on the parameters of the problem. The explicit expressions for~$P_{ij}$ are available through symbolic computation and are generally long and cumbersome. For conciseness, we will provide explicit expressions only for the coefficients whose form and dependence on parameters are used in the stability analysis.

The polynomial~$P(\lambda)$ has five roots. Four roots~$\{\lambda_j\}_{j=1}^4$ are of order~$\mathcal{O}(\sqrt\mu)$, 
\begin{subequations}\label{eq:lambda_14}
\begin{equation}
    \lambda_j=\sqrt{\mu} \,r_j+o(\sqrt\mu),\quad j=1,2,\cdots,4,
\end{equation}
where~$\{r_j\}_{j=1}^4$ are the roots of
\begin{equation}\label{eq:eq4ri}
\frac{((\mathcal{R}_1-\mathcal{R}_2)\gamma_{12}+\mathcal{R}_2\eta_2\gamma_2)((\mathcal{R}_1 - \mathcal{R}_2)\gamma_2 + \gamma_{12}\mathcal{R}_2)\sigma_{12}}{\gamma_{12}}r^4 + P_{10}r^2 + P_{20}=0,
\end{equation}
\end{subequations}
and an additional root that satisfies
\[
\lambda_5=-\frac{\gamma_2(\mathcal{R}_1 - \mathcal{R}_2)+ \gamma_{12}\mathcal{R}_2}{\mathcal{R}_1}+\mathcal{O}(\mu).
\]
The latter root is negative in the feasible parameter space, i.e., when~$(\mathcal{R}_1,\mathcal{R}_2)\in\Omegamu$.  

We now consider the roots~\eqref{eq:lambda_14}. To simplify the stability criteria, we group the parameters into constants~$A$ and~$B$ defined as
\begin{equation}\label{eq:AnB}
\begin{split}
A&=\gamma_1\gamma_{12}(\mathcal{R}_2+\sigma_{12}(\mathcal{R}_1 - \mathcal{R}_2))(\hat{\mathcal{R}}_2^1-1)\mathcal{R}_1,
\\B&=\frac{\gamma_1\sigma_{12}( \eta_2\gamma_2\mathcal{R}_2+\gamma_{12}(\mathcal{R}_1 - \mathcal{R}_2))(\gamma_{12}\mathcal{R}_2+\gamma_2(\mathcal{R}_1- \mathcal{R}_2))(1 + \sigma_{12}(\mathcal{R}_1 - 1))\mathcal{R}_1R_1^{CE}}{\gamma_{12}(1+\sigma_{12}\mathcal{R}_1 R_1^{CE})}
,
\end{split}
\end{equation}
so that the coefficients of~\eqref{eq:eq4ri} read as
\[\begin{split}
P_{10}&=A+B,\\ P_{20}&=\frac{A\,B}{(\gamma_2(\mathcal{R}_1-\mathcal{R}_2) + \gamma_{12}\mathcal{R}_2)(\sigma_{12}(\mathcal{R}_1-\mathcal{R}_2) + \mathcal{R}_2)}.
\end{split}
\]
For~$(\mathcal{R}_1,\mathcal{R}_2)\in\Omegamu$,~$A>0$ and~$B>0$. Thus,~$P_{10}>0$ and~$P_{20}>0$. The discriminant,~$\Delta$, of the quadratic equation~\eqref{eq:eq4ri} equals
\[
\begin{split}
\Delta&=P_{10}^2 - 4\frac{((\mathcal{R}_1-\mathcal{R}_2)\gamma_{12}+\mathcal{R}_2\eta_2\gamma_2)((\mathcal{R}_1 - \mathcal{R}_2)\gamma_2 + \gamma_{12}\mathcal{R}_2)\sigma_{12}}{\gamma_{12}}P_{20}\\&=(A-B)^2+(\gamma_{12}-\eta_2\gamma_2\sigma_{12})(\hat{\mathcal{R}}_2^1-1)C\\
&=(A-B)^2+\gamma_{12}(1-q_{12}^c)(\hat{\mathcal{R}}_2^1-1)C,
\end{split}
\]
where
\[
C=\frac{4\sigma_{12}\gamma_1^2(\gamma_{12}\mathcal{R}_2+\gamma_2(\mathcal{R}_1-\mathcal{R}_2))(\eta_2\gamma_2\mathcal{R}_2+ \gamma_{12}(\mathcal{R}_1-\mathcal{R}_2))(1+\sigma_{12}(\mathcal{R}_1-1))\mathcal{R}_1^2\mathcal{R}_2R_1^{CE}}{1+\sigma_{12}\mathcal{R}_1R_1^{CE}}>0,
\]
and
\begin{equation}\label{eq:sigma12c}
    q_{12}^c=\frac{\sigma_{12}\eta_2\gamma_2}{\gamma_{12}}.
\end{equation} 
The analysis gives rise to a new quantity,~$q_{12}^c$. 

Biologically,~$q_{12}^c$ quantifies the relative transmission advantage of secondary infections. This advantage can stem from heightened infectivity ($\eta_2$), decreased crossed-immunity or enhanced susceptibility ($\sigma_{12}$), or a longer infectious period ($\gamma_{12}^{-1}$). Crucially, this formulation reveals that instability does not strictly require weak cross-immunity ($\sigma_{12} \approx 1$) as~\cite{gavish2024newoscillatoryregimetwostrain} implied. Rather, instability emerges whenever the combined epidemiological traits of a secondary infection confer an overall transmission advantage ($q_{12}^c \ge 1$).
  
\begin{remark}
The invasion number~$\hat{\mathcal{R}}_2^1$, see~\eqref{eq:Rinvasion1}, equals
    \begin{equation}\label{eq:Rinvasion1_sigma12c}
\hat{\mathcal{R}}_2^1=\frac{\mathcal{R}_2}{\mathcal{R}_1}\left[1+q_{12}^c(\mathcal{R}_1-1)\right]+\mathcal{O}(\mu),
\end{equation}
so that the ratio~$q_{12}^c$ relates to the contribution of secondary infections to the invasion number. 
\end{remark}

In the case~$q_{12}^c<1$, it follows that~$\Delta>0$, and hence~$r_j^2<0$ for~$j=1,2,3,4$. This implies that the relevant eigenvalues are purely imaginary at leading order, and stability is determined by the next-order terms. Numerous works, e.g.,~\cite{andreasen1997dynamics, gomes2002dynamics, dawes2002onset, chung2016dynamics}, have derived these terms, revealing long and cumbersome expressions from which stability could only be determined in highly specific parameter cases. We will not pursue this direction here but will present numerical results suggesting that the coexistence steady-state is stable in this parameter regime.

When~$q_{12}^c>1$, the second term of~$\Delta$ is negative, and therefore~$\Delta$ can become negative if~$A$ is sufficiently close to~$B$. The following lemma points to a single curve in the parameter space along which~$A=B$.

\begin{lemma}\label{lem:curveGamma}
Let~$\Omegamu$ be defined by~\eqref{eq:omega_mu}, and let~$A$ and~$B$ be defined by~\eqref{eq:AnB}. Then, for a sufficiently small~$\mu>0$, there exists a single curve~$\Gamma^*\subset\Omegamu$ along which~$A=B$.
\end{lemma}
\begin{proof}
See Appendix~\ref{app:proof_lem_curveGamma}.
\end{proof}
Following Lemma~\ref{lem:curveGamma}, if~$q_{12}^c>1$, then~$\Delta<0$ along~$\Gamma^*$. In this case, at least one of the roots~$\{r_j\}_{j=1}^4$ has a positive real part, and therefore the coexistence steady-state is unstable along~$\Gamma^*$.

In the borderline case,~$q_{12}^c=1$,~$\Delta=0$ only when~$A=B$. Here, stability is determined by the next-order terms. Expressions~\eqref{eq:AnB} for~$A$ and~$B$ reduce to
\[
A=(\mathcal{R}_1 - \mathcal{R}_2)(\sigma_{12}\mathcal{R}_1+1 - \sigma_{12}),\quad B=[\gamma_{12} + \eta_2\gamma_2\sigma_{12}(\mathcal{R}_1 - 1)]\mathcal{R}_2 - \gamma_{12}\mathcal{R}_1,
\]
so that
\[
P_{10}=\gamma_1(\mathcal{R}_1 - \mathcal{R}_2) + \gamma_2(\mathcal{R}_2 - 1),\quad P_{20}=\gamma_1\gamma_2(\mathcal{R}_1 - \mathcal{R}_2)(\mathcal{R}_2 - 1).
\]
Following~\cite{gavish2024newoscillatoryregimetwostrain}, to ascertain stability when~$\Delta=0$, we consider the higher-order expansion
\begin{equation}\label{eq:expansion_with_si_base}
\lambda_j=\sqrt{\mu} \,r_j +\mu^{3/4} s_j+\mu q_j+\mathcal{O}(\mu\sqrt[4]\mu),\quad j=1,2,3,4.
\end{equation}
Substituting~\eqref{eq:expansion_with_si_base} into~$P(\lambda)$ and equating powers of~$\mu$ yields, at leading order,
\[
r_{1,2}=\pm i\sqrt{\gamma_2(\mathcal{R}_2-1)},\quad r_{3,4}=\pm i\sqrt{\gamma_1\mathcal{R}_1-\gamma_2\mathcal{R}_2}.
\]
At the next order, we have
\begin{equation}\label{eq:sqrtDeltasj}
\Delta r_js_j=0,\quad j=1,2,3,4.
\end{equation}
Given~$\Delta=0$, the subsequent order yields
    \begin{equation}\label{eq:si}
    s_j^2=-\frac{r_j}2\frac{P_{11}r_j^2 + P_{21}}{P_{10}}.
\end{equation}
Note that~$s_j^2$ is purely imaginary. Consequently, 
\[
s_j=c_j\,(\pm 1+i).
\]
Hence, two of the eigenvalues~$\{\lambda_j\}_{j=1}^4$ have a positive real part, and thus the coexistence steady-state is unstable along~$\Gamma^*$ and in its vicinity. 

The following Proposition summarizes the analytical results attained in this section.
\begin{proposition}[Instability along~$\Gamma^*$]\label{prop:instablity_along_Gamma*}
    Let the relative transmission advantage of a secondary infection over a primary infection satisfy
    \begin{equation}\label{eq:sigma12condition}
    q_{12}^c\ge1,
    \end{equation} and let~$\mu>0$ be sufficiently small. For any parameter set such that~$(\mathcal{R}_1,\mathcal{R}_2)\in\Gamma^*$, where the curve~$\Gamma^*\subset\Omegamu$ is defined in Lemma~\ref{lem:curveGamma}, the unique coexistence steady state,~$\phi^{CE}$, of system~\eqref{eq:base_model} is linearly unstable.
\end{proposition}

\begin{remark}
The analysis in~\cite{gavish2024newoscillatoryregimetwostrain} examines the stability of~$\phi^{CE}$ with respect to the dynamics of the more general model~\eqref{eq:model} in the regime~$0 < \sigma_{21} \le \sigma_{12} \le 1$. In retrospect, the restrictive assumptions of that study caused condition~\eqref{eq:sigma12condition} to simplify to~$\sigma_{12} \ge 1$, meaning the analysis was effectively constrained to identifying instability only in the mathematical borderline case~$q_{12}^c=1$ with~$\sigma_{12} =1$. Consequently, the current analysis serves as an analytical refinement and generalization of~\cite{gavish2024newoscillatoryregimetwostrain}. By relaxing the prior restricted assumptions, we demonstrate that the previously identified instability is not merely an isolated consequence of a mathematical borderline, but rather a specific instance of a much broader and biologically robust oscillatory regime. Furthermore, the fact that the condition~$q_{12}^c=1$ remains applicable in~\cite{gavish2024newoscillatoryregimetwostrain} for any~$0 < \sigma_{21} \le 1$, rather than being restricted to the limit~$\sigma_{21} \to 0$, suggests that the current analysis extends beyond the base model~\eqref{eq:base_model}. See also Section~\ref{sec:general_two_strain}.
\end{remark}

\section{Numerical Study}\label{sec:numeric_base_model}
The analysis in Section~\ref{sec:stability_phiCE_base} shows that when secondary infections possess a relative transmission advantage over primary infections,~$q_{12}^c\ge1$, there exists a parameter region in which the coexistence steady-state,~$\phi^{CE}$, is unstable, and that this region includes the curve~$\Gamma^*$. However, the analysis does not address the dynamics of~$\phi^{CE}$ beyond this characterized region. Therefore, it is not clear whether the steady state is stable outside this regime or if other regions of instability exist. Furthermore, our analytical results do not provide information on the global dynamics of the system when~$\phi^{CE}$ is unstable, for example, whether its solutions converge to a periodic limit cycle or display more complicated dynamics. In this section, we complement the analysis with a numerical study to obtain a clearer picture of the system's global behavior.  Numerical details are provided in Appendix~\ref{app:numericalDetails}.

\begin{figure}[ht!]
\begin{center}
\includegraphics[width=\textwidth]{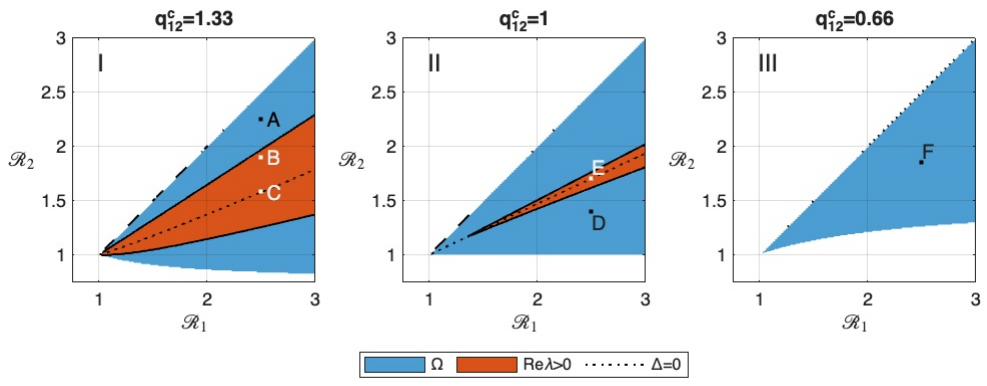}
\caption{Bifurcation diagram for the base model~\eqref{eq:base_model} as a function of~$\mathcal{R}_1$ and~$\mathcal{R}_2$. The parameters~$\beta_1$ and~$\beta_2$ vary while the other parameters are kept fixed: $\gamma_1=1$,~$\gamma_2=1.1$,~$\sigma_{12}=1.1$, $\eta_2=1.1$, $\mu=2.5\cdot10^{-4}$ and  
I: ~$q_{12}^c=1.33$.  II:~$q_{12}^c=1$. III:~$q_{12}^c=0.66$. The diagram shows the region~$\Omegamu$ where the coexistence steady-state exists (blue) and the region where it is unstable (red). The curve~$\Gamma^*$ is shown as a dotted black curve. Points A-F are at~$\mathcal{R}_1=2.5$ with the following $\mathcal{R}_2$ values: A(2.25), B(1.9), C(1.58), D(1.4), E(1.7), and F(1.85).}
\label{fig:bifurcation_diagram_base_alternative}
\end{center}
\end{figure}
Figure~\ref{fig:bifurcation_diagram_base_alternative} presents a bifurcation diagram for the base model~\eqref{eq:base_model} across several values of~$q_{12}^c$, the relative transmission advantage of a secondary infection. In all cases, we set~$\gamma_1=1$ and~$\mu=2.5\cdot 10^{-4}$; thus, assuming a time unit of one week, the average recovery time from a primary infection with strain 1 is one week, and the mean host lifespan is roughly 76 years. While Proposition~\ref{prop:instablity_along_Gamma*} formally establishes instability for sufficiently small~$\mu$, this parameter choice acts to confirm that the predicted oscillatory behavior robustly arises under biologically realistic values of~$\mu$ and~$\gamma_i$.  

To construct this scenario, we choose parameters that confer a slight transmission advantage to secondary infections: heightened infectivity ($\eta_2=1.1$) and enhanced susceptibility ($\sigma_{12}=1.1$). However, the overall transmission advantage,~$q_{12}^c$, is ultimately determined by the duration of the secondary infectious period,~$\gamma_{12}=\eta_2\sigma_{12}\gamma_{2}/q_{12}^c$. For our chosen values, setting~$q_{12}^c \in \{1.33, 1, 0.66\}$ yields secondary recovery rates of~$\gamma_{12} \approx 1, 1.33,$ and $2$, respectively, corresponding to recovery periods of $7-14$ days. Furthermore, exploring other parameter combinations, such as neutral infectivity ($\eta_2=1$) and partial cross-immunity ($\sigma_{12}<1$), produces no qualitative change in the system's dynamics, see Appendix~\ref{app:numerical_examples}.

Figure~\ref{fig:bifurcation_diagram_base_alternative}I shows the case~$q_{12}^c=1.33$, where the secondary transmission advantage satisfies condition~\eqref{eq:sigma12condition}. Within the coexistence region~$\Omega_\mu$, the steady state is computed and its linear stability is evaluated via the spectral radius of the corresponding Jacobian,~$\max \mathrm{Re}\lambda$, with stable regimes colored blue and unstable regimes colored red. As Proposition~\ref{prop:instablity_along_Gamma*} implies, a distinct region of instability (red) emerges, with the critical curve~$\Gamma^*$ (black dotted curve) lying safely within its interior.

Figure~\ref{fig:bifurcation_diagram_base_alternative}III presents the case~$q_{12}^c=0.66$, for which condition~\eqref{eq:sigma12condition} is not satisfied. We observe that~$\phi^{CE}$ is stable whenever it exists, providing numerical evidence that the next-order terms maintain the stability of the coexistence equilibrium in this regime.  Finally, Figure~\ref{fig:bifurcation_diagram_base_alternative}II shows the borderline case~$q_{12}^c=1$. In accordance with Proposition~\ref{prop:instablity_along_Gamma*}, we observe that~$\phi^{CE}$ is unstable in a narrow regime the vicinity of the curve~$\Gamma^*$. 
\begin{figure}[ht!]
\begin{center}
\includegraphics[width=\textwidth]{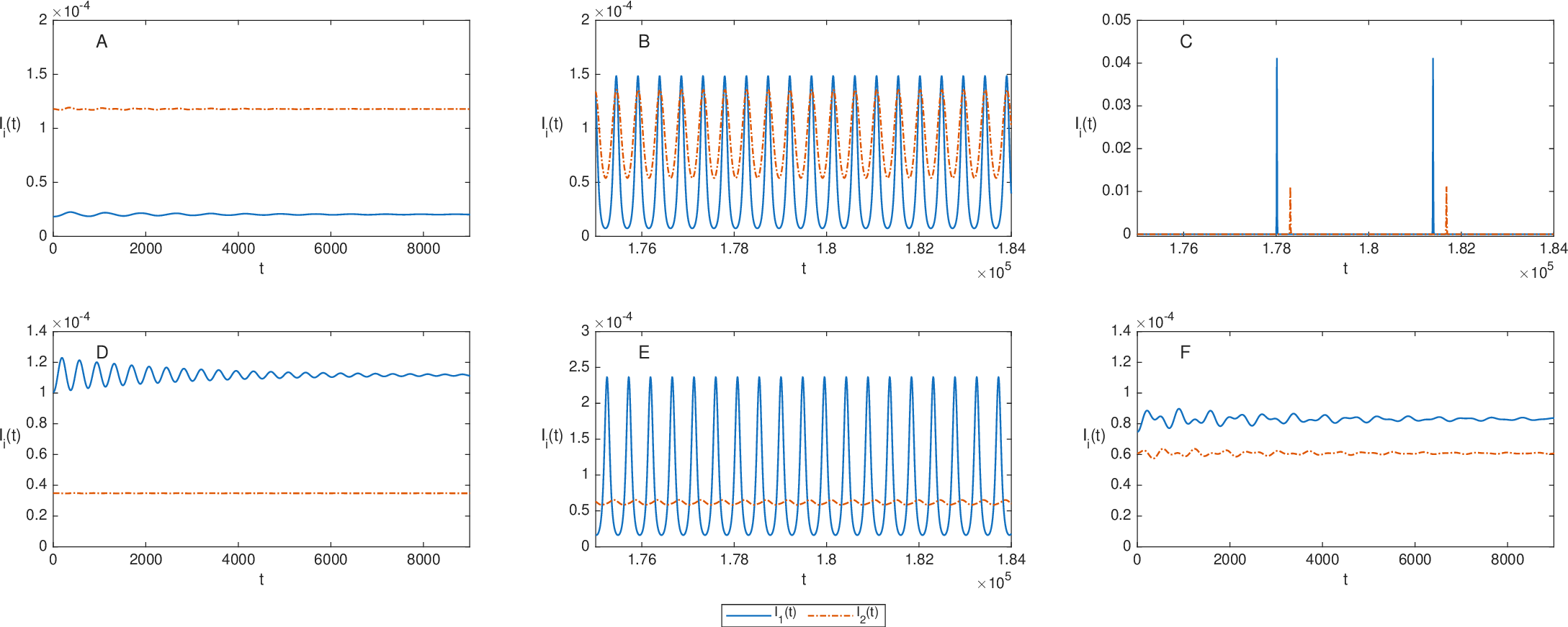}
\caption{Solutions~$I_1(t)$ and~$I_2(t)$ of~\eqref{eq:base_model} corresponding to points A-F in Figure~\ref{fig:bifurcation_diagram_base_alternative}.}
\label{fig:SolutionsSample_alternative}
\end{center}
\end{figure}

To visualize the system's dynamics across different parameter regimes, Figures~\ref{fig:SolutionsSample_alternative}A-F present the temporal solutions of~\eqref{eq:base_model} corresponding to points A-F in Figure~\ref{fig:bifurcation_diagram_base_alternative}, all initialized near coexistence ($\phi_0\approx \phi^{CE}$). Points A, D, and F reside in the stability region of~$\phi^{CE}$; as expected, the corresponding solutions naturally converge to this steady state (Figures~\ref{fig:SolutionsSample_alternative}A,~\ref{fig:SolutionsSample_alternative}D, and~\ref{fig:SolutionsSample_alternative}F, respectively).  Conversely, when~$\phi^{CE}$ is unstable, all steady states of the bounded system are unstable, precluding convergence to an equilibrium and resulting in sustained oscillatory or complex dynamics. Points B and E lie near the boundary of this instability region (Figures~\ref{fig:bifurcation_diagram_base_alternative}I and~\ref{fig:bifurcation_diagram_base_alternative}II, respectively). We observe that the solutions exhibit small-amplitude oscillations centered around~$\phi^{CE}$, consistent with the local Hopf-like findings of previous research~\cite{chung2016dynamics,gavish2024newoscillatoryregimetwostrain}. Surprisingly, at point C, located well within the instability region, we observe a drastically distinct dynamic. The solution consists of short, intense surges of strain 1, followed by a surge of strain 2, and a subsequent low-prevalence time interval before a new cycle begins.  The numerical investigation in the next subsection will focus on these unexpected oscillations, and show that they can coexist with the small-amplitude Hopf-like limit cycles.

The parameters chosen for the graphs in this section correspond to 
heightened infectivity and enhanced susceptibility ($\eta_2=1.1$ and~$\sigma_{12}=1.1$) of secondary infections compared to a primary infection, whereas the overall relative transmission advantage or disadvantage of a secondary infection over a primary infection,~$q_{12}^c$, determines the duration of the infectious period,~$\gamma_{12}=\eta_2\sigma_{12}\gamma_{2}/q_{12}^c$.  To demonstrate that the key parameter determining the stability behavior of~$\phi^{CE}$ is~$q_{12}^c$, rather than the parameters that it depends on, we present in Appendix~\ref{app:numerical_examples} an additional numerical example with parameters that reflect neutral infectivity and partial cross-immunity ($\sigma_{12}=0.8$ and $\eta_2=1$).  As expected, the qualitative results are the same.

\subsection{Observation of an additional oscillatory mechanism}
The oscillatory solution in Figure~\ref{fig:SolutionsSample_alternative}B exhibits a small-amplitude oscillation, whereas the solution in Figure~\ref{fig:SolutionsSample_alternative}C demonstrates large-amplitude, outbreak-like behavior. To establish whether these distinct behaviors stem from coexisting attractors at the exact same parameter choices, we investigate the system's sensitivity to initial conditions. Figure~\ref{fig:pointB_alternative} demonstrates a striking bistability for the parameters corresponding to point B in Figure~\ref{fig:bifurcation_diagram_base_alternative}I ($\mathcal{R}_1 = 2.5, \mathcal{R}_2 = 1.9$). When initialized in close proximity to the coexistence equilibrium $\phi^{CE}$ via a minor perturbation, explicitly constructed by setting $I_1(0) \approx 0.9 I_1^{CE}$ and scale-normalizing the remaining compartments to preserve the unit population invariant $\sum \phi_i = 1$, see~\eqref{eq:I0nearCE}, the trajectory approaches a low-amplitude limit cycle enclosing $\phi^{CE}$ (Figure~\ref{fig:pointB_alternative}A, and the red dash-dotted curve in Figure~\ref{fig:pointB_alternative}C). Conversely, under the identical parameter set, initiating the system near the single-strain equilibrium $\phi^{EE,2}$ by introducing a trace invader population $I_1(0) = 10^{-5}$ (with the remaining non-zero compartments of $\phi^{EE,2}$ scale-normalized accordingly) uncovers a distinct, large-amplitude oscillatory attractor (Figure~\ref{fig:pointB_alternative}B, and the gray curve in Figure~\ref{fig:pointB_alternative}C). Numerical continuation shows that this bistability persists under small parameter perturbations and is not a fine-tuned numerical artifact. As discussed in Section~8, in a companion paper, we will demonstrate that this distinct, large-amplitude oscillatory mechanism is driven by a Shilnikov-like orbit connected to a single-strain equilibrium.

\begin{figure}[ht!]
\begin{center}
\includegraphics[width=\textwidth]{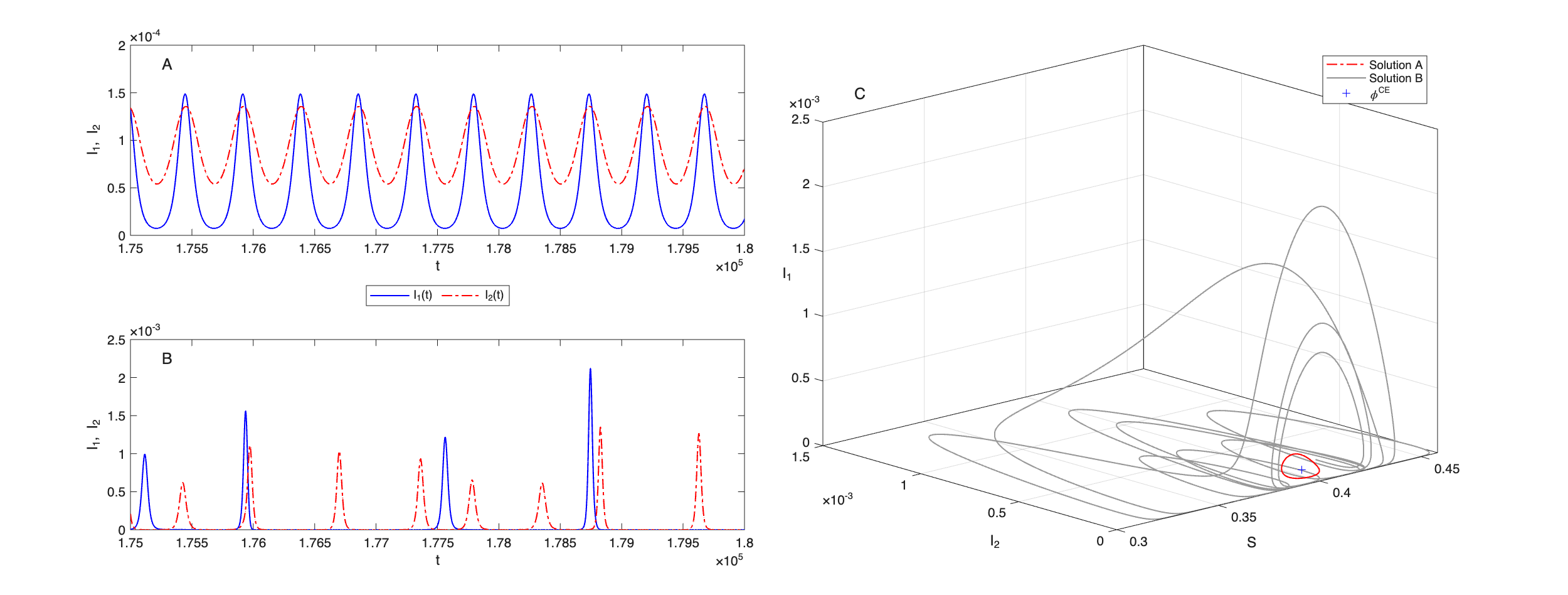}
\caption{Bistability and coexisting attractors for the base model~\eqref{eq:base_model} evaluated at point B in Figure~\ref{fig:bifurcation_diagram_base_alternative}I ($\mathcal{R}_1=2.5, \mathcal{R}_2=1.9$, $\gamma_1=1, \gamma_2=1.1, \sigma_{12}=1.1, \eta_2=1.1, \mu=2.5\times 10^{-4}$). (A) Small-amplitude local limit cycle initialized near coexistence with $I_1(0) \approx 0.9 I_1^{CE}$, see~\eqref{eq:I0nearCE}. (B) Large-amplitude outbreak cycle under identical parameters initialized near the single-strain equilibrium with an invasion seed $I_1(0) = 10^{-5}$. (C) Three-dimensional phase space trajectory in the $(S, I_1, I_2)$ subspace showcasing the coexisting local limit cycle (dash-dotted red) and the macroscopic attractor (gray) around $\phi^{CE}$ (blue cross). All trajectories are simulated via \texttt{ode45} and plotted over the stationary time interval $[T - 5/\mu, T]$ with $T = 5\times 10^4$.}
\label{fig:pointB_alternative}
\end{center}
\end{figure}

To systematically map how these different oscillatory behaviors are distributed across the parameter space, we monitor the spike amplitude for solutions of~\eqref{eq:base_model} along the line~$\mathcal{R}_1=2.5$ where the initial condition is near~$\phi^{CE}$, specifically we take
\begin{equation}\label{eq:I0nearCE}
I_1(0)=\frac{0.9\cdot I_1^{CE}}{1-0.1\cdot I_1^{CE}}
\end{equation}
while the rest of the variables remain the same as~$\phi^{CE}$, up to normalization, $\phi_0=\frac{\phi^{CE}}{1-0.1\cdot I_1^{CE}}$.
\begin{figure}[ht!]
\begin{center}
\includegraphics[width=\textwidth]{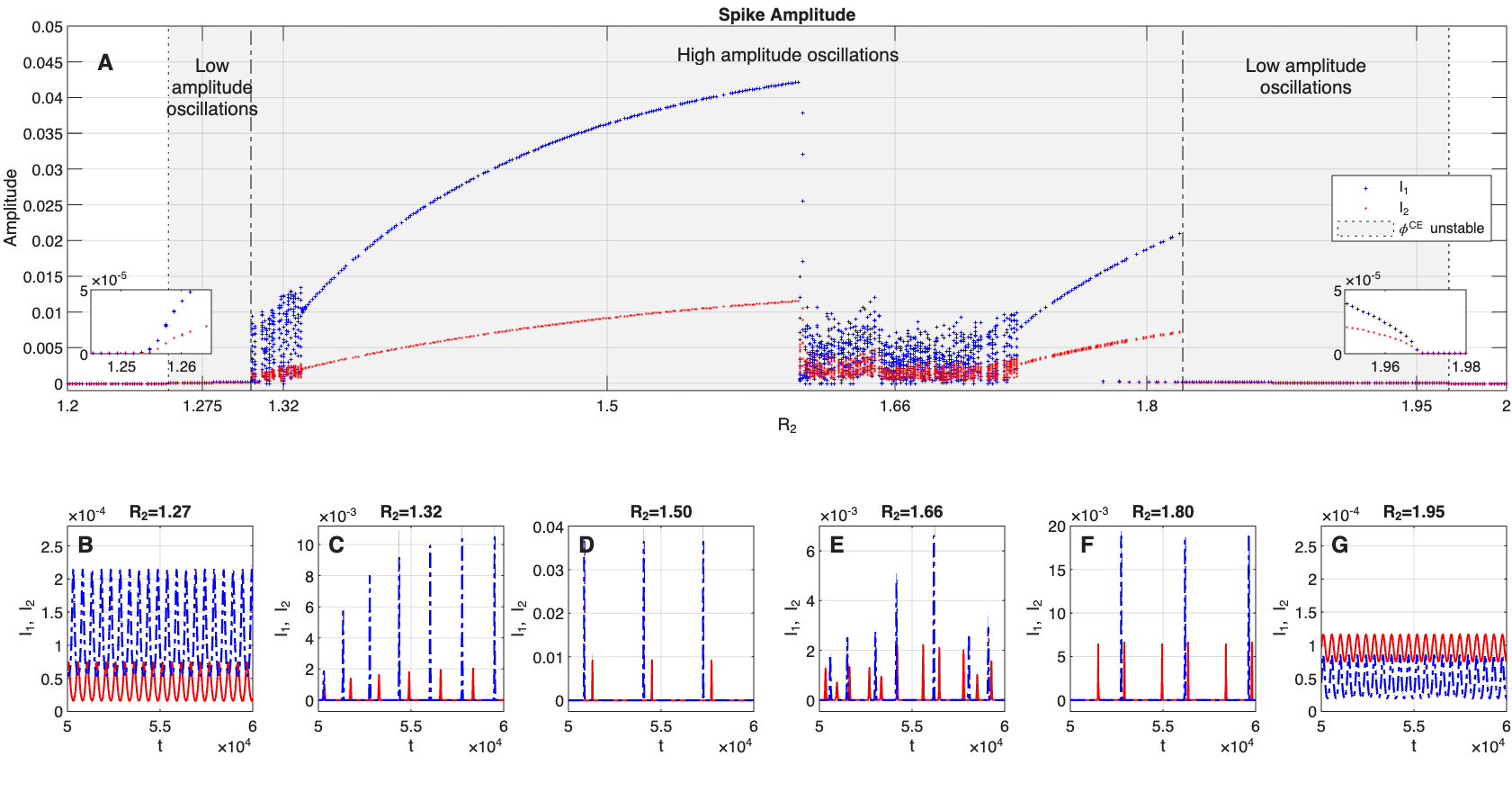}
\caption{\textbf{Bifurcation analysis and temporal dynamics with respect to parameter $\mathcal{R}_2$.} 
    The top row displays one-parameter bifurcation diagrams showing the variation in spike amplitude (A) as a function of~$\mathcal{R}_2$ where the initial condition is near~$\phi^{CE}$.  The blue crosses denote the maxima for variable~$I_1$, while the red dots denote~$I_2$. The shaded gray area denote the region where~$\phi^{CE}$ is unstable.  Near the boundaries of this region, we observe low amplitude oscillations, see inset graphs for zoom-in around the boundaries, while in the interior of the instability region we observe high amplitude oscillations.
    The bottom row presents solution dynamics for specific values of~$\mathcal{R}_2$.}
\label{fig:ISI_amplitude_graph_alternative}
\end{center}
\end{figure}

The region of instability of~$\phi^{CE}$ is marked by a shaded gray area in Figure~\ref{fig:ISI_amplitude_graph_alternative}A.  As expected, in this region, we observe that the solution is oscillatory.
We observe that the system exhibits small-amplitude oscillations near the boundaries, see inset graphs in Figure~\ref{fig:ISI_amplitude_graph_alternative} for a zoom-in on the boundaries of the instability region. For instance, 
Figure~\ref{fig:ISI_amplitude_graph_alternative}B provides an example of this behavior near the lower boundary of the instability region ($\mathcal{R}_2=1.27$), while Figure~\ref{fig:ISI_amplitude_graph_alternative}G illustrates a similar small-amplitude pattern near the upper boundary ($\mathcal{R}_2=1.95$). However, the one-parameter bifurcation diagram in Figure~\ref{fig:ISI_amplitude_graph_alternative}A  reveals that, for the above initial conditions, these Hopf-like limit cycles do not extend to the full instability region. Throughout the vast majority of the instability region's interior, the system solutions transition into large-amplitude, outbreak-like oscillations, see Figures~\ref{fig:ISI_amplitude_graph_alternative}C-F.  Particularly, within the region of large-amplitude oscillations, we observe two sub-regions,~$[1.33,1.6]$ and $[1.72,1.82]$, where the solution peaks reach a unique value.  These regions correspond to regular and periodic large-amplitude oscillations, as observed in Figures~\ref{fig:ISI_amplitude_graph_alternative}D and~\ref{fig:ISI_amplitude_graph_alternative}F.  The complementary regions correspond to more complex oscillatory patterns, possibly aperiodic, as observed in Figures~\ref{fig:ISI_amplitude_graph_alternative}C and E.

\section{Analytical Tractability and Structural Robustness within the~$\sigma_{21}=0$ Family}\label{sec:extensions_numerical}
The overarching goal of this study is to better understand oscillatory phenomena in multi-strain epidemic systems through an analytically tractable family of models. To demonstrate this tractability, we show that the coexistence equilibrium,~$\phi^{CE}$, can be explicitly computed even for a highly complex generalization of the baseline model. Analyzing this extended formulation also allows us to establish the structural robustness of our primary results. Specifically, by incorporating realistic epidemiological conditions, such as waning immunity and the isolation of primarily infected individuals, we confirm that the analytically derived instability regime and oscillatory behavior persist well beyond the base formulation. The subsequent section demonstrates that this structural robustness extends even to models outside the~$\sigma_{21}=0$ family.

In the following we consider an extended model with waning immunity that includes two additional isolation compartments, $Q_i$, for those infected with strain $i$. To distinguish between recovery histories, we also introduce a compartment $R_2$ for individuals recovered from a primary infection with strain 2, while reserving compartment $R$ exclusively for those recovered from a secondary infection. The schematic diagram of this extended model is shown in Figure~\ref{fig:diagram_extended}.
\begin{figure}[ht!]
\begin{center}
\includegraphics[width=0.75\textwidth]{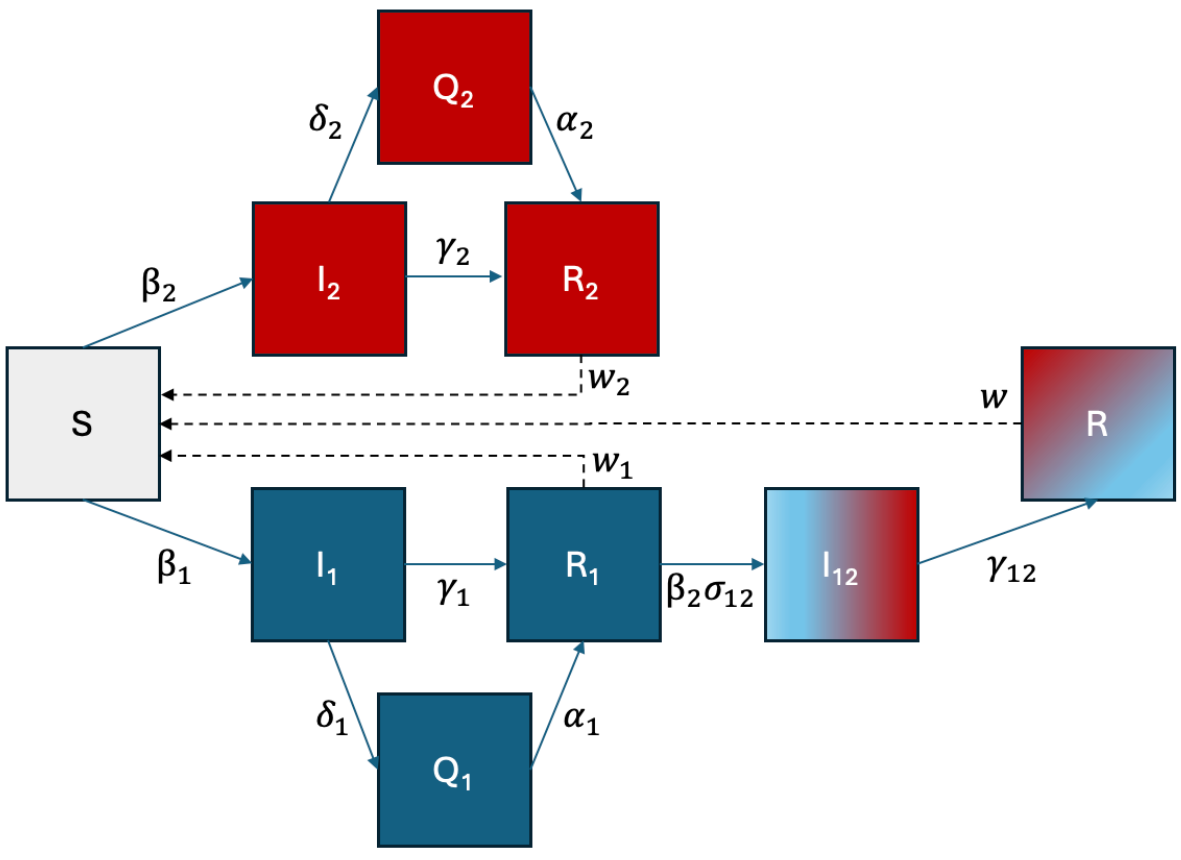}
\caption{Schematic diagram of disease dynamics for the extended model. This model includes compartments for quarantine ($Q_i$) and distinguishes between recovery from primary ($R_2$) and secondary ($R$) infections with strain 2.}
\label{fig:diagram_extended}
\end{center}
\end{figure} 

The extended model is given by the following system of equations:
\begin{subequations} \label{eq:extended_model}
\begin{align}
        \frac{\text{d}S}{\text{d}t} &=\mu(1-S)-\beta_1 I_1S-\beta_2(I_2+\eta_2I_{12})S+w_1R_1+w_2R_2+wR,\label{eq:full_eqS}\\[1ex]
         \frac{\text{d}I_1}{\text{d}t}  &= \beta_1I_1S-(\mu+\gamma_1+\delta_1)I_1,\label{eq:full_eqI1}\\[1ex]
                  \frac{\text{d}Q_1}{\text{d}t}  &= \delta_1I_1-(\mu+\alpha_1)Q_1,\label{eq:full_eqQ1}\\[1ex]
         \frac{\text{d}I_2}{\text{d}t}  &= \beta_2(I_2+\eta_2I_{12})S-(\mu+\gamma_2+\delta_2)I_2,\label{eq:full_eqI2}\\[1ex]
         \frac{\text{d}Q_2}{\text{d}t}  &= \delta_2I_2-(\mu+\alpha_2)Q_2,\label{eq:full_eqQ2}\\
         \frac{\text{d}I_{12}}{\text{d}t}  &= \sigma_{12}\beta_2(I_2+\eta_2I_{12})R_1-(\mu+\gamma_{12})I_{12},\label{eq:full_eqI12}\\[1ex]        
         \frac{\text{d}{R_1}}{\text{d}t} &= \gamma_1I_1+\alpha_1 Q_1-\sigma_{12}\beta_2(I_2+\eta_2I_{12})R_1-(\mu+w_1) R_1,\label{eq:full_eqR1}\\[1ex]
                  \frac{\text{d}{R_2}}{\text{d}t} &= \gamma_2I_2+\alpha_2 Q_2-(\mu+w_2) R_2,\label{eq:full_eqR2}\\[1ex]
         \frac{\text{d}{R}}{\text{d}t} &=\gamma_{12}I_{12}-(w+\mu) R,\label{eq:full_eqR}
\end{align}
\end{subequations}
where~$\delta_i$ is the rate at which individuals infected with strain~$i$ are isolated, and~$\alpha_i$ is the rate at which they recover and are released from isolation. The rate at which immunity wanes following a primary infection with strain~$i$ is denoted by~$w_i$, and the waning rate after a secondary infection is denoted by~$w$. Note that when~$w_2=w$, the dynamics of those recovered from an infection with strain 2 are independent of whether it was a primary or secondary infection; in such cases, the compartments~$R_2$ and~$R$ can be combined, reducing to a structure closer to that of model~\eqref{eq:base_model}. 

Similar to the base model, we consider non-negative initial compartment sizes that sum to one. For the extended model parameters, we require the following conditions:
\begin{equation}\label{eq:parameterspace_extended}
\alpha_i>0,\quad \delta_i\ge0, \quad w_i\ge0,\quad w\ge0,\qquad i=1,2,
\end{equation}
in addition to the baseline feasible space defined in~\eqref{eq:parameterspace_reduced_only}. The basic reproduction number of~\eqref{eq:extended_model} reads as
\[
\mathcal{R}_0=\max\{\mathcal{R}_1, \mathcal{R}_2\},
\]
where the strain-specific reproduction numbers are now updated to account for isolation:
\begin{equation}\label{eq:RRi_extended}
\mathcal{R}_i=\frac{\beta_i}{\gamma_i+\delta_i+\mu}.
\end{equation}
\subsection{Analytic tractability}
Similar to the base model, we begin by establishing the single-strain endemic equilibria. If~$\mathcal{R}_i>1$, system~\eqref{eq:extended_model} admits a single-strain steady state~$\phi^{EE,i}$. The compartment sizes of~$\phi^{EE,1}$ are:
\begin{equation}\label{eq:phiEE1_extended}
\begin{split}
    S^{EE,1}&=\frac1{\mathcal{R}_1},\quad I_1^{EE,1}=\frac{(\mathcal{R}_1 - 1)(\alpha_1+\mu)(w_1+\mu)}{\mu^2+(\gamma_1+w_1 + \alpha_1 + \delta_1 )\mu+(\gamma_1+\delta_1+w_1)\alpha_1+\delta_1w_1}\frac{1}{\mathcal{R}_1},\\ Q_1^{EE,1}&=\frac{\delta_1}{\alpha_1+\mu}I_1^{EE,1},\quad R_1^{EE,1}=\frac{(\delta_1 + \gamma_1)\alpha_1 + \gamma_1\mu}{(\alpha_1+\mu)(w_1+\mu)}I_1^{EE,1},\\ I_2^{EE,1}&=I_{12}^{EE,1}=R_2^{EE,1}=Q_2^{EE,1}=0,
\end{split}
\end{equation}
and the compartment sizes of~$\phi^{EE,2}$ are:
\begin{equation}\label{eq:phiEE2_extended}
\begin{split}
    S^{EE,2}&=\frac1{\mathcal{R}_2},\quad I_2^{EE,2}=\frac{(\mathcal{R}_2 - 1)(\alpha_2+\mu)(w_2+\mu)}{\mu^2+(\gamma_2+w_2 + \alpha_2 + \delta_2 )\mu+(\gamma_2+\delta_2+w_2)\alpha_2+\delta_2w_2}\frac{1}{\mathcal{R}_2},\\ Q_2^{EE,2}&=\frac{\delta_2}{\alpha_2+\mu}I_2^{EE,2},\quad R_2^{EE,2}=\frac{(\delta_2 + \gamma_2)\alpha_2 + \gamma_2\mu}{(\alpha_2+\mu)(w_2+\mu)}I_2^{EE,2},\\ I_1^{EE,2}&=I_{12}^{EE,2}=R_1^{EE,2}=Q_1^{EE,2}=0.
\end{split}
\end{equation}

To determine the stability of these boundary steady states and the conditions for coexistence, we define the invasion reproduction numbers for the extended model. 
The invasion number of strain 2 into the endemic equilibrium of strain 1 is given by:
\begin{equation}\label{eq:invasion_R21_extended}
\hat{\mathcal{R}}_2^1 = \frac{\mathcal{R}_2}{\mathcal{R}_1}\left[1+\frac{\sigma_{12}\eta_2(\gamma_2+\mu+\delta_2)\mathcal{R}_1}{\gamma_{12}+\mu}R_1^{EE,1}\right].
\end{equation}
Conversely, because strain 2 does not generate secondary infections by strain 1 ($\sigma_{21}=0$), the invasion number of strain 1 into the endemic equilibrium of strain 2 reduces simply to:
\begin{equation}\label{eq:invasion_R12_extended}
\hat{\mathcal{R}}_1^2 = \frac{\mathcal{R}_1}{\mathcal{R}_2}.
\end{equation}

Using these invasion numbers, we define the extended coexistence region,~$\Omega_{ext}$, as the parameter space where both strains can mutually invade:
\begin{equation}\label{eq:omega_mu_extended}
\Omega_{ext} = \{ (\mathcal{R}_1,\mathcal{R}_2) \mid \hat{\mathcal{R}}_2^1 > 1,\quad \hat{\mathcal{R}}_1^2 > 1 \}.
\end{equation}

With the steady states and invasion boundaries established, the stability and coexistence conditions follow directly:

\begin{proposition}[Stability of Single-Strain Steady States]\label{prop:stability_EE_extended}
Let~$\mathcal{R}_i>1$ such that the single-strain endemic steady state~$\phi^{EE,i}$, of system~\eqref{eq:extended_model} exists.  Then,~$\phi^{EE,i}$ is linearly stable if and only if the invasion number,~$\hat{\mathcal{R}}_j^i$, of the competing strain~$j\ne i$ is stricly less than one:
\[
\hat{\mathcal{R}}_j^i<1.
\]
\end{proposition}
\begin{proof}
    See Appendix~\ref{app:proof_phiEE_extended}.
\end{proof}

\begin{proposition}[Coexistence Steady State]\label{prop:phiCE_extended}
   Let~$(\mathcal{R}_1,\mathcal{R}_2)\in \Omega_{ext}$. Then, system~\eqref{eq:extended_model} possesses a unique coexistence steady state,~$\phi^{CE}$, characterized by strictly positive infectious populations~$I_1>0$ and~$I_2>0$. The compartment sizes of~$\phi^{CE}$ are given by:
\begin{subequations}
\begin{equation}
    \begin{split}
S^{CE}&=\frac1{\mathcal{R}_1},\quad I_1^{CE}=\frac{(\mu + \alpha_1)(\mu+w_1+\sigma_{12}(\mu + \gamma_2 + \delta_2)\mathcal{R}_1I_2^{CE})}{(\delta_1 + \gamma_1)\alpha_1 + \gamma_1\mu}R_1^{CE},\\ 
I_2^{CE}&=(\mu + w)(\mu + w_2)(\mu + \gamma_{12})(\mu + \alpha_2)(\hat{\mathcal{R}}_2^1-1)\times\\&\frac{\alpha_1\delta_1 + \alpha_1\gamma_1 + \alpha_1\mu + \alpha_1w_1 + \delta_1\mu + \delta_1w_1 + \gamma_1\mu + \mu^2 + \mu w_1}{a+b\,(\mathcal{R}_1-\mathcal{R}_2)},\\
I_{12}^{CE}&=\frac{\mathcal{R}_1-\mathcal{R}_2}{\eta_2\mathcal{R}_2}I_2^{CE},\quad Q_1^{CE}=\frac{\delta_1}{\mu+\alpha_1}I_1^{CE},\quad Q_2^{CE}=\frac{\delta_2}{\mu+\alpha_2}I_2^{CE},\\ R_2^{CE}&=\frac{(\delta_2 + \gamma_2)\alpha_2 + \mu\gamma_2}{(\mu + \alpha_2)(\mu + w_2)}I_2^{CE},\\
R^{CE}&=\frac{\gamma_{12}c}{\mu+w}I_2^{CE},\quad R_1^{CE} = \frac{(\mu + \gamma_{12})(\mathcal{R}_1-\mathcal{R}_2)}{(\mu + \gamma_2 + \delta_2)\eta_2\sigma_{12}\mathcal{R}_1\mathcal{R}_2},
\end{split}
\end{equation}
where
\begin{equation}\begin{split}
a&=(\mu^2 + (\delta_2 + \gamma_2 + w_2 + \alpha_2)\mu + (\delta_2 + \gamma_2 + w_2)\alpha_2 + w_2\delta_2)\times\\&((\delta_1 + \gamma_1)\alpha_1 + \gamma_1\mu)(\mu + \gamma_2 + \delta_2)(\mu + w)\mathcal{R}_2\eta_2\sigma_{12},\\ 
b&=(\mu + w_2)(\mu + \gamma_2 + \delta_2)(\mu + \alpha_2)\sigma_{12}\times\\&\left[\mu^3 + (w + \gamma_{12} + \alpha_1 + \delta_1 + \gamma_1)\mu^2+ ((w + \delta_1 + \gamma_1)\gamma_{12} + w(\delta_1 + \gamma_1))\alpha_1 + w\delta_1\gamma_{12}\right.\\& \left.+((w + \delta_1 + \gamma_1 + \gamma_{12})\alpha_1 + (w + \delta_1 + \gamma_1)\gamma_{12} + w(\delta_1 + \gamma_1))\mu \right].    
\end{split}
\end{equation}
\end{subequations}
\end{proposition}
\begin{proof}
See Appendix~\ref{app:proof_phiCE_extended}.
\end{proof}
Proposition~\ref{prop:phiCE_extended} demonstrates how the core assumption, the elimination of secondary infections by strain 1 ($\sigma_{21}=0$), preserves the mathematical tractability of the model by yielding an explicit expression for the coexistence steady state even in a highly complex model featuring quarantine and waning immunity.

\subsection{Structural robustness}
The extended model~\eqref{eq:extended_model} differs from the base formulation by incorporating two distinct mechanisms for susceptible replenishment: a slow demographic turnover, as in the base model, and a significantly faster mechanism driven by the waning of immunity. This addition alters the asymptotic structure underpinning the analytical results in Section~\ref{sec:stability_phiCE_base}. Consequently, it is not {\it a priori} clear whether the extended model will preserve the oscillatory dynamics of the base system. While a full mathematical analysis of this extended case is beyond the scope of the current paper, we present a numerical analysis demonstrating that the system's dynamics exhibit remarkable structural robustness. To evaluate this robustness, we consider the extended model under a biologically plausible parameter set:
\begin{equation}\label{eq:parameter_values_extended}
w=0.012,\quad w_1=0.008,\quad w_2=0.01, \quad \delta_1=0.1,\quad\delta_2=0.05,\quad\alpha_1=1.2,\quad\alpha_2=0.8.
\end{equation}
Assuming a time unit of one week, this configuration represents realistic recovery and isolation periods of approximately one week, alongside immunity waning over a timescale of roughly two years. These values also reflect slight, realistic asymmetries between the two strains. As in the numerical example of Figure~\ref{fig:bifurcation_diagram_base_alternative}, we determine~$\gamma_{12}$ according to the value of the relative transmission advantage.

In the base model~\eqref{eq:base_model}, the analysis identified the parameter $q_{12}^c$ as the exact relative transmission advantage determining the onset of limit-cycle instabilities. Lacking a closed-form transcendental expression for the exact bifurcation threshold in the higher-dimensional framework, we introduce a structural analogue, $q_{12}^{c,\text{extended}}$, to serve as an interpretive heuristic and parameter-scaling tool. This heuristic leverages the property that the threshold scales with the ratio of secondary to primary infectious period durations weighted by susceptibility and infectivity adjustments. Transposing this relation to the extended framework with quarantine, we define:
\begin{equation}\label{eq:q12_extended}
q_{12}^{c,\text{extended}} = \frac{\eta_2 \sigma_{12} (\gamma_2 + \delta_2)}{\gamma_{12}}.
\end{equation}
We emphasize that $q_{12}^{c,\text{extended}}$ is not used as an approximate classification metric for our stability boundaries. Because Proposition~5 yields an explicit expression for the coexistence steady state $\phi^{CE}$, the linear stability of the system is evaluated exactly and directly by computing the eigenvalues of the full $9\times9$ Jacobian matrix across our parameter sweeps. Instead, we employ~\eqref{eq:q12_extended} as a scaling rule to explicitly isolate and calibrate the secondary recovery rate:
\begin{equation}\label{eq:gamma12_mapping}
\gamma_{12} = \frac{\eta_2 \sigma_{12} (\gamma_2 + \delta_2)}{q_{12}^{c,\text{extended}}}.
\end{equation}
This allows us to evaluate the structural robustness of the instability mechanism under comparable transmission conditions across different versions of the model.

\begin{figure}[ht!]
\begin{center}
\includegraphics[width=\textwidth]{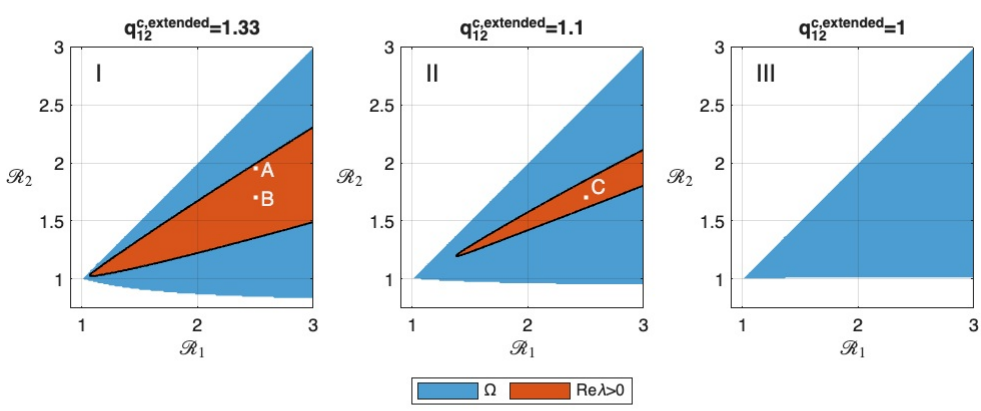}
\caption{Bifurcation diagram for the extended model~\eqref{eq:extended_model} as a function of~$\mathcal{R}_1$ and~$\mathcal{R}_2$. The parameters~$\beta_1$ and~$\beta_2$ vary while the other parameters are kept fixed and are given by~\eqref{eq:parameter_values_extended}. The diagram shows the region~$\Omegamu$ where the coexistence steady-state exists (blue) and where it is unstable (red). Points A-C are marked within the oscillatory region for analysis.}
\label{fig:bifurcationDiagram_gamma12_02_extended}
\end{center}
\end{figure}
Figure~\ref{fig:bifurcationDiagram_gamma12_02_extended}I presents the bifurcation diagram in the~$(\mathcal{R}_1,\mathcal{R}_2)$ plane for~$q_{12}^{c, {\rm extended}}=1.33$. As expected, we observe a substantial oscillatory region, resembling Figure~\ref{fig:bifurcation_diagram_base_alternative}I. Within this region, we mark two distinct points, A and B, which are dynamically analogous to points B and C in the base model (Figure~\ref{fig:bifurcation_diagram_base_alternative}), respectively. As in the base model, point A exhibits a local periodic limit cycle about~$\phi^{CE}$,  see Figure~\ref{fig:solutionSample_extended}A, while point B gives rise to larger-amplitude outbreaks interspersed with low-activity periods, see Figure~\ref{fig:solutionSample_extended}B. 
\begin{figure}[ht!]
\begin{center}
\includegraphics[width=\textwidth]{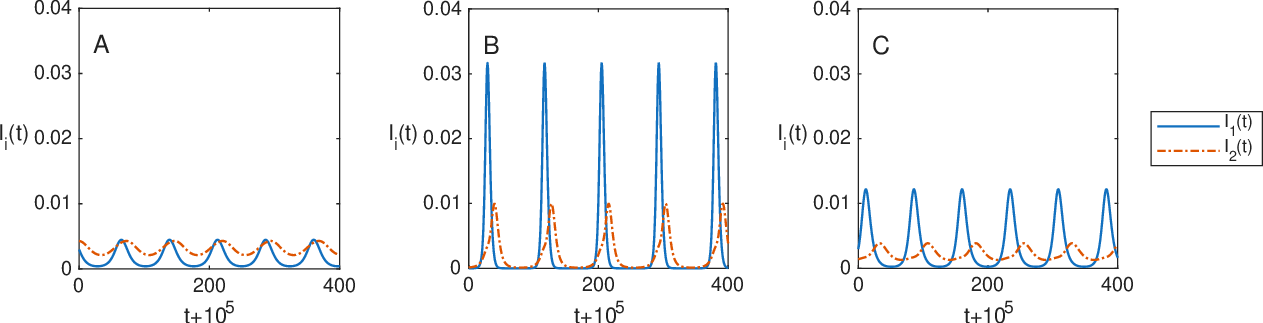}
\caption{Solutions~$I_1(t)$ and~$I_2(t)$ for the extended model~\eqref{eq:extended_model} at points A, B, and C from Figure~\ref{fig:bifurcationDiagram_gamma12_02_extended}.}
\label{fig:solutionSample_extended}
\end{center}
\end{figure}

However, the dynamical differences between the small-amplitude (point A) and large-amplitude (point B) oscillations are significantly less pronounced in the extended model. The amplitude ratio between these two solutions is approximately 3, a stark contrast to the ratio of roughly 200 observed between the analogous solutions in the base model. Moreover, the extended model exhibits much shorter low-activity periods, which are comparable in length to the duration of the high-activity outbreaks. Furthermore, bifurcation analysis reveals that the oscillatory solutions of the extended model remain entirely regular, see Figure~\ref{fig:ISI_amplitude_graph_extended}, whereas the base model exhibited both regular and irregular oscillations, see Figure~\ref{fig:ISI_amplitude_graph_alternative}. It is highly plausible that the faster replenishment of susceptible in the extended model acts to dampen extreme fluctuations, thereby contributing to these observed changes in amplitude scaling, temporal clustering, and regularity.
\begin{figure}[ht!]
\begin{center}
\includegraphics[width=\textwidth]{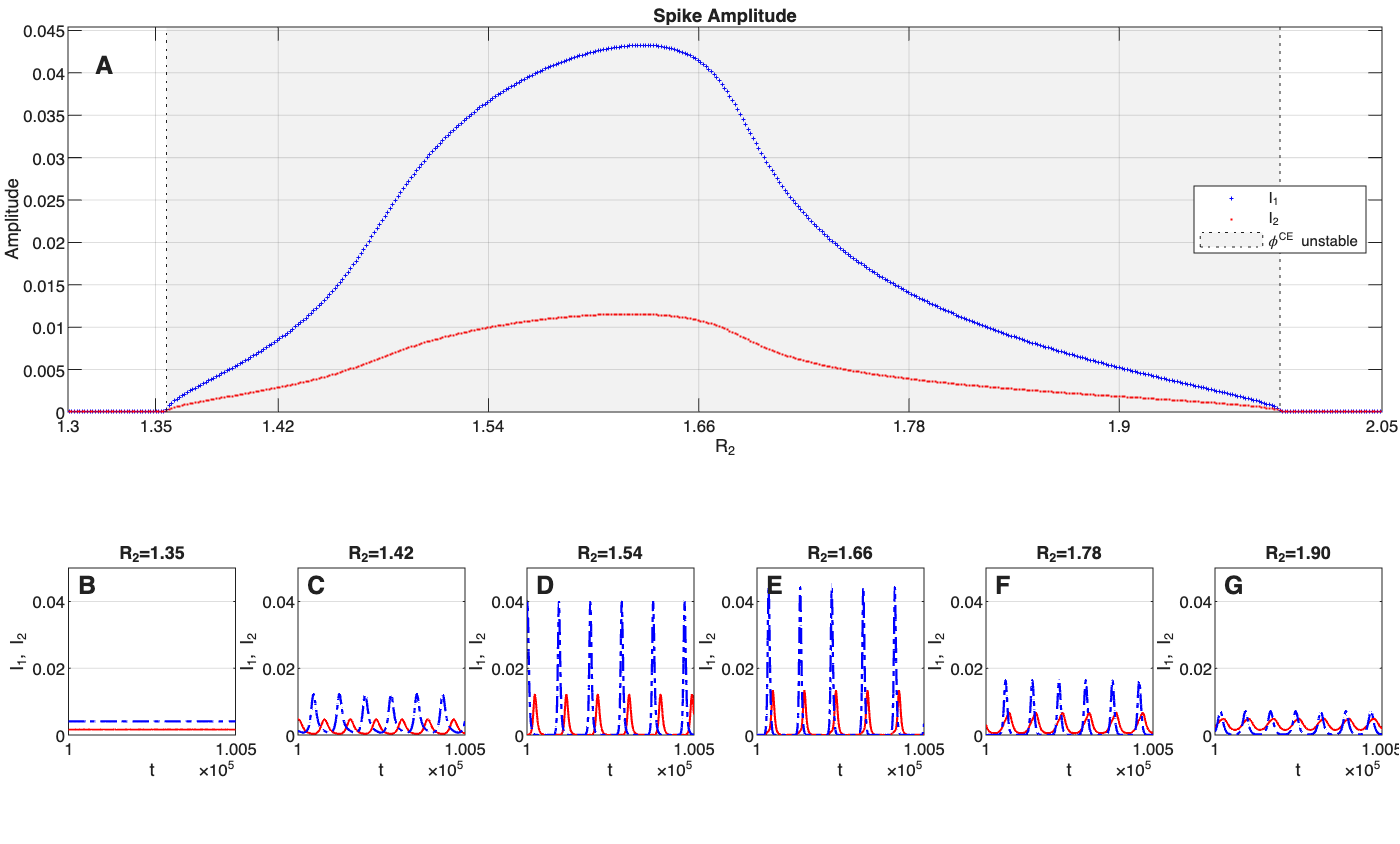}
\caption{\textbf{Bifurcation analysis and temporal dynamics with respect to parameter $\mathcal{R}_2$.} 
    The top row displays one-parameter bifurcation diagrams showing the variation in spike amplitude (A) as a function of~$\mathcal{R}_2$ evaluated along the line~$\mathcal{R}_1=2.5$, corresponding to the parameter slice containing the dynamics of Figure~\ref{fig:bifurcationDiagram_gamma12_02_extended}I. The initial condition is set near~$\phi^{CE}$. The blue crosses denote the maxima for variable~$I_1$, while the red dots denote the maxima for~$I_2$. The shaded gray area denotes the region where~$\phi^{CE}$ is unstable. 
    The bottom row presents solution dynamics for specific values of~$\mathcal{R}_2$.}
\label{fig:ISI_amplitude_graph_extended}
\end{center}
\end{figure}
Figure~\ref{fig:bifurcationDiagram_gamma12_02_extended}II presents the exact numerical bifurcation diagram evaluated for $q_{12}^{c,\text{extended}} = 1.1$. We observe that this choice represents a borderline scenario, dynamically analogous to the base model threshold behavior shown in Figure~\ref{fig:bifurcation_diagram_base_alternative}II. Conversely, when setting $q_{12}^{c,\text{extended}} = 1.0$ (Figure~\ref{fig:bifurcationDiagram_gamma12_02_extended}III), the unstable oscillatory region is completely absent from the feasible parameter space. 

Based on these exact numerical stability computations, we find that for the chosen parameter configuration~\eqref{eq:parameter_values_extended}, the transition to instability occurs at an empirical threshold value satisfying $1.0 < q_{12}^{c,\text{extended}} \le 1.1$. This conditional observation confirms that while $q_{12}^{c,\text{extended}}$ functions as a reliable and highly accurate heuristic for scaling parameters and identifying the structural conditions that spark instabilities, the fast susceptible replenishment introduced by waning immunity slightly shifts the true analytical boundary away from unity.

In summary, despite the differences  between the base model and the extended model, both share key structural features regarding their oscillatory behavior. Specifically, both frameworks capture the emergence of oscillations beyond the critical threshold quantity and demonstrate distinct regimes of low-amplitude versus high-amplitude oscillations. Crucially, the explicit availability of the coexistence steady-state~$\phi^{CE}$, see Proposition~\ref{prop:phiCE_extended}, establishes a pathway for subsequent analytic studies. While further research is necessary to refine our understanding of the complex impact of waning immunity, the foundational approach taken in this work has proven highly fruitful, yielding critical insights into the underlying mechanisms that drive epidemic cycles.

\section{An outlook on Structural Robustness Beyond the $\sigma_{21}=0$ Family}\label{sec:general_two_strain}
This work aims to better understand oscillatory phenomena in multi-strain epidemic systems by focusing on an analytically tractable family of models where secondary infections with one strain are excluded ($\sigma_{21}=0$). This approach proved highly successful, revealing the critical role of $q_{12}^c$, the $\Gamma^*$ curve, and the limiting case in which the oscillatory domain is confined to the vicinity of $\Gamma^*$. Furthermore, pinpointing the domain where small-amplitude oscillations arise enables the identification of additional oscillatory mechanisms and provides clear directions for future research. To test the potential breadth of these findings without exceeding the scope of the current work, we conduct an initial numerical exploration to determine whether these dynamics extend to cases where~$\sigma_{21}\ne0$.

Figure~\ref{fig:bifurcation_diagram_base_nonzero_sigma21} presents bifurcation diagrams for model~\eqref{eq:model} as a function of~$\mathcal{R}_1$ and~$\mathcal{R}_2$ across three distinct cases. Because this model allows for secondary infections from both strains, we define~$q_{21}^c$ as the analogue to~$q_{12}^c$~\eqref{eq:sigma12c}. This parameter represents the relative transmission advantage of a secondary infection with strain one over a primary infection with strain two:
$$
q_{21}^c=\frac{\sigma_{21}\eta_1\gamma_1}{\gamma_{21}}.
$$

When $\sigma_{21} \neq 0$, the coexistence equilibrium $\phi^{CE}$ is no longer available in a fully decoupled closed-form vector like in the specialized $\sigma_{21}=0$ family. However, the steady-state equations can still be reduced exactly to a cubic polynomial in terms of the susceptible population $S^*$~\cite{gavish2024newoscillatoryregimetwostrain}. At each parameter coordinate, the roots of this cubic polynomial are computed algebraically using standard matrix companion methods via MATLAB's \texttt{roots} routine. Feasibility and uniqueness are strictly enforced by filtering for a real root that lies within the domain $0 < S^* < \min(1/\mathcal{R}_1, 1/\mathcal{R}_2)$. Once this unique value of $S^*$ is established, the remaining epidemiological compartments are evaluated sequentially via explicit closed-form algebraic expressions  of $S^*$. 

The first case explores the dynamics of cocirculating strains driven by partial, symmetric cross-immunity ($\sigma_{12} = \sigma_{21} < 1$). Such a scenario holds strong biological relevance and has been used to model diseases such as influenza, see, e.g.,~\cite{andreasen1997dynamics}. Here, primary infections maintain a relative transmission advantage over secondary infections, analogous to the scenario presented in Figure~\ref{fig:bifurcation_diagram_base_alternative}A. As expected, we observe that the coexistence steady-state~$\phi^{CE}$ is stable whenever it exists, see Figure~\ref{fig:bifurcation_diagram_base_nonzero_sigma21}A. 

The second case explores an asymmetric scenario where only secondary infections with strain one possess a relative transmission advantage over primary infections with strain two ($q_{12}^c < 1 < q_{21}^c$). Similar to the case in Figure~\ref{fig:bifurcation_diagram_base_alternative}B where~$q_{12}^c > 1$ and~$q_{21}^c = 0$, we observe a parameter region where~$\phi^{CE}$ becomes unstable, indicated by the red region in Figure~\ref{fig:bifurcation_diagram_base_nonzero_sigma21}B. To systematically map the distribution of oscillatory behaviors across the parameter space, we monitor the spike amplitudes for solutions of model~\eqref{eq:model} along the line~$\mathcal{R}_1 = 2.5$, where the initial condition is near~$\phi^{CE}$, see Figure~\ref{fig:ISI_amplitude_graph_base_example_caseB}A. Resembling the approach and dynamics of Figure~\ref{fig:ISI_amplitude_graph_alternative}A, the system exhibits regions of low-amplitude oscillations (e.g., Figures~\ref{fig:ISI_amplitude_graph_base_example_caseB}B--D), as well as both regular and irregular high-amplitude oscillations (Figures~\ref{fig:ISI_amplitude_graph_base_example_caseB}E and~\ref{fig:ISI_amplitude_graph_base_example_caseB}F).

Finally, we consider dynamics induced by antibody-dependent enhancement (ADE), following the framework of Ferguson et al.~\cite{ferguson1999effect}.  In this case, secondary infections hold a relative transmission advantage over primary infections. Crucially, while this case includes symmetric strain parameters, the system still exhibits robust oscillatory dynamics (see the red region in Figure~\ref{fig:bifurcation_diagram_base_nonzero_sigma21}C). This demonstrates that these oscillations are fundamentally driven by the relative advantage of secondary infections rather than parametric asymmetry. Indeed, the bifurcation diagram reveals that the oscillatory region spans a segment of the~$\mathcal{R}_1=\mathcal{R}_2$ line, where the two strains are completely symmetric. However, unlike all other cases considered in this work, this oscillatory region lies strictly along the boundary of the existence region for~$\phi^{CE}$. 
Evaluating the corresponding spike amplitudes (Figure~\ref{fig:ISI_amplitude_graph_base_example_caseC}A), we observe a continued diversity in the oscillatory regimes: regular and irregular high-amplitude solutions (Figures~\ref{fig:ISI_amplitude_graph_base_example_caseC}B--D), alongside distinct low-amplitude oscillations (Figures~\ref{fig:ISI_amplitude_graph_base_example_caseC}E--G). 
It is important to contextualize these findings relative to diseases like dengue fever. While ADE is notably linked to dengue fever~\cite{ferguson1999effect}, model~\eqref{eq:model} does not capture the more intricate immunological complexities, such as temporary cross-immunity, associated with the disease~\cite{aguiar2022mathematical}.  

\begin{figure}[ht!]
\begin{center}
\includegraphics[width=\textwidth]{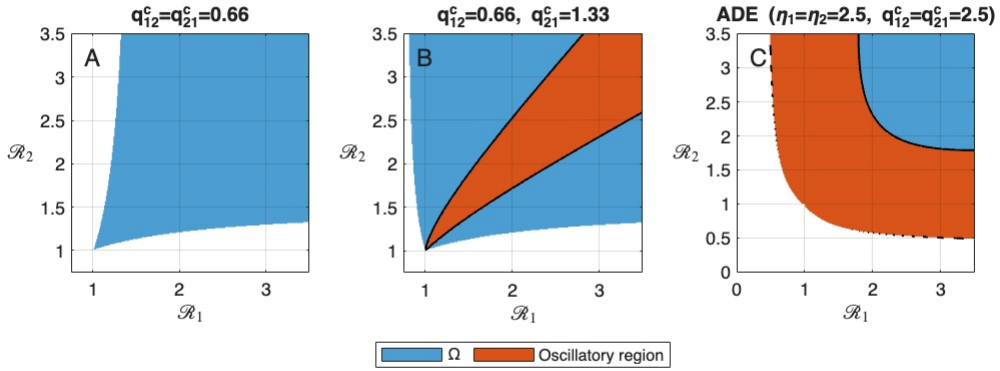}
\caption{Bifurcation diagram for the model~\eqref{eq:model} as a function of~$\mathcal{R}_1$ and~$\mathcal{R}_2$. The parameters~$\beta_1$ and~$\beta_2$ vary while the other parameters are kept fixed.  
A: $\gamma_1=1$,~$\gamma_2=1.1$,~$\sigma_{12}=1.1$, $\eta_2=1.1$, $\mu=2.5\cdot10^{-4}$ and  
I: ~$q_{12}^c=1.125$.  II:~$q_{12}^c=1$. III:~$q_{12}^c=0.875$. The diagram shows the region~$\Omegamu$ where the coexistence steady-state exists (blue) and the region where it is unstable (red). The curve~$\Gamma^*$ is shown as a dotted black curve. Points A-F are at~$\mathcal{R}_1=2.5$ with the following $\mathcal{R}_2$ values: A(2.25), B(1.85), C(1.65), D(1.4), E(1.7), and F(1.85).}
\label{fig:bifurcation_diagram_base_nonzero_sigma21}
\end{center}
\end{figure}

\begin{figure}[ht!]
\begin{center}
\includegraphics[width=\textwidth]{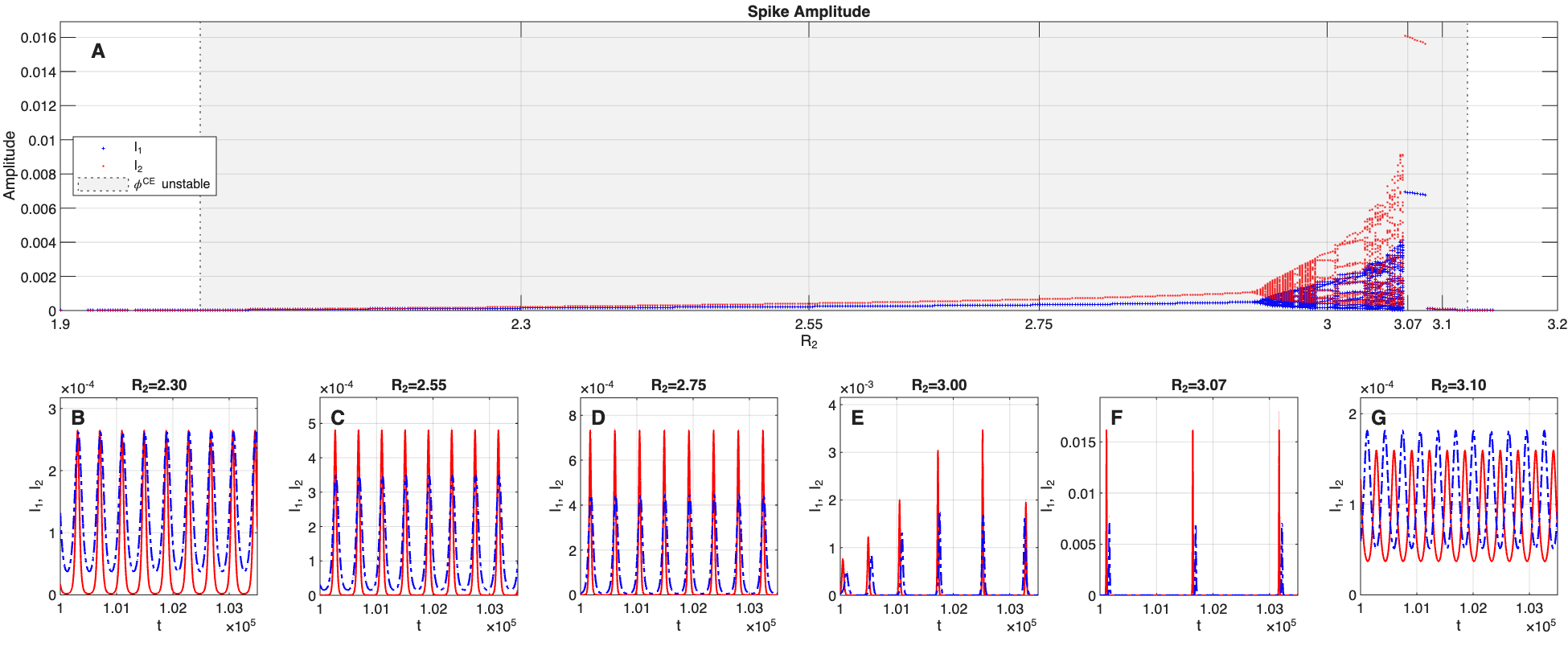}
\caption{\textbf{Bifurcation analysis and temporal dynamics with respect to parameter $\mathcal{R}_2$.} 
    The top row displays one-parameter bifurcation diagrams showing the variation in spike amplitude (A) as a function of~$\mathcal{R}_2$ evaluated along the line~$\mathcal{R}_1=2.5$, corresponding to the parameter slice containing the dynamics of Figure~\ref{fig:bifurcation_diagram_base_nonzero_sigma21}B. The initial condition is set near~$\phi^{CE}$. The blue crosses denote the maxima for variable~$I_1$, while the red dots denote the maxima for~$I_2$. The shaded gray area denotes the region where~$\phi^{CE}$ is unstable. 
    The bottom row presents solution dynamics for specific values of~$\mathcal{R}_2$.}
\label{fig:ISI_amplitude_graph_base_example_caseB}
\end{center}
\end{figure}

\begin{figure}[ht!]
\begin{center}
\includegraphics[width=\textwidth]{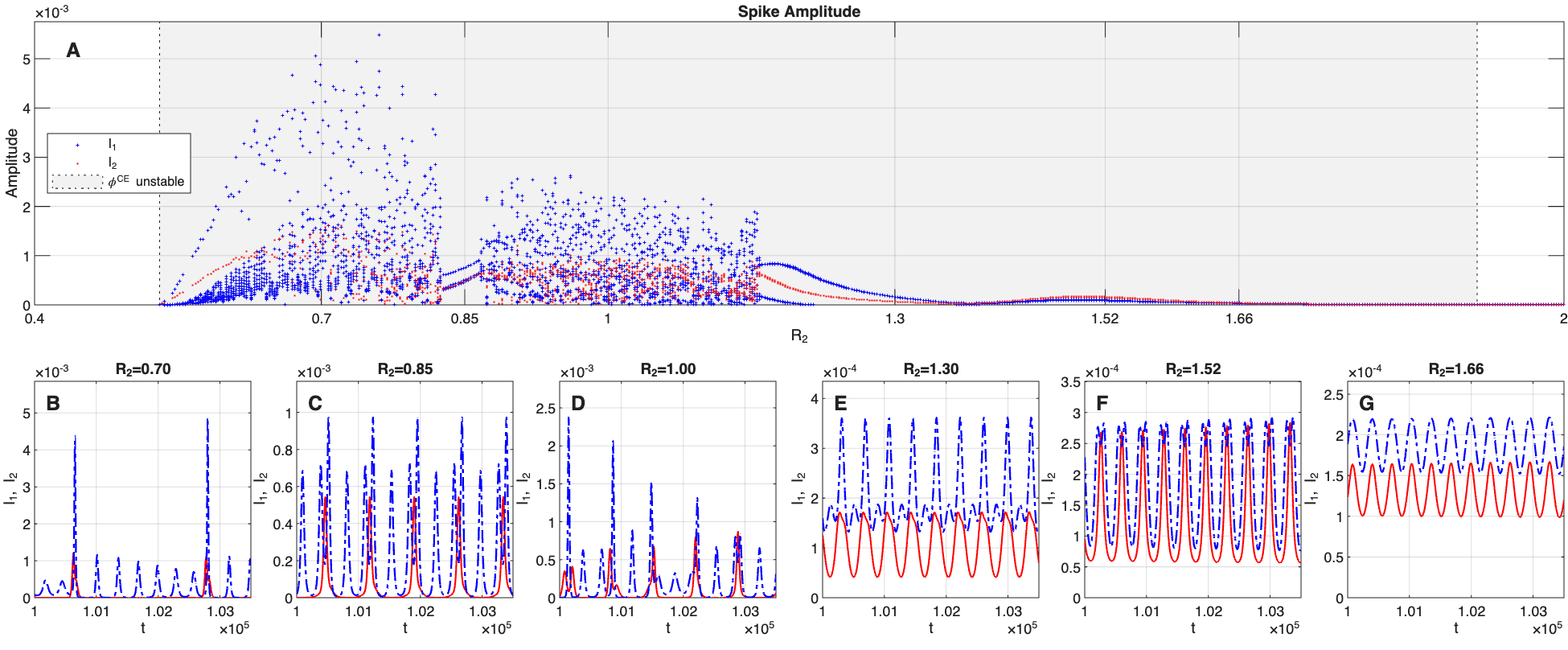}
\caption{\textbf{Bifurcation analysis and temporal dynamics with respect to parameter $\mathcal{R}_2$.} 
    The top row displays one-parameter bifurcation diagrams showing the variation in spike amplitude (A) as a function of~$\mathcal{R}_2$ evaluated along the line~$\mathcal{R}_1=2.5$, corresponding to the parameter slice containing the dynamics of Figure~\ref{fig:bifurcation_diagram_base_nonzero_sigma21}C. The initial condition is set near~$\phi^{CE}$. The blue crosses denote the maxima for variable~$I_1$, while the red dots denote the maxima for~$I_2$. The shaded gray area denotes the region where~$\phi^{CE}$ is unstable. 
    The bottom row presents solution dynamics for specific values of~$\mathcal{R}_2$.}
\label{fig:ISI_amplitude_graph_base_example_caseC}
\end{center}
\end{figure}

Ultimately, exploring parameter regimes where~$\sigma_{21} \neq 0$ reveals that the core oscillatory phenomena identified in our analytically tractable base model are highly robust. Even though the coexistence steady-state~$\phi^{CE}$ is no longer explicitly solvable in this extended setting, numerical investigations confirm the persistence of key structural features: the threshold-driven onset of oscillations and the clear bifurcation into low- and high-amplitude regimes. Although further research is necessary to fully unravel the complex interactions in these extended regimes, the foundational approach taken in this work has proven highly fruitful. By analyzing the tractable base model, we have established a robust framework that yields critical insights into the underlying mechanisms driving multi-strain epidemic cycles.

\section{Discussion}\label{sec:discussion}
This study introduces a family of analytically tractable two-strain models defined by the parameter choice~$\sigma_{21}=0$. This restriction bypasses the algebraic complexity typically associated with the coexistence equilibrium, enabling a rigorous analysis of oscillatory phenomena even when incorporating features like antibody-dependent enhancement (ADE) ($\eta_i>1$) and asymmetric recovery rates.  This approach showed that  the relative transmission advantage of a secondary infection is the primary driver of limit-cycle oscillations.  The results provide new context to the results in~\cite{gavish2024newoscillatoryregimetwostrain}. In retrospect, that study evaluated a borderline scenario residing precisely at the onset of instability, explaining why the previously identified oscillatory regime was so parametrically narrow. Moving away from this mathematical borderline, we uncover parameter regions that support far more robust oscillations, dynamic behaviors that likely better reflect real biological conditions.

While mapping these small-amplitude oscillations, we encountered an unexpected phenomenon: the emergence of distinct, large-amplitude outbreak cycles, even when the initial condition is near the coexistence equilibrium. These macroscopic oscillations operate independently of the local Hopf bifurcation. Under certain initial conditions, they dominate the phase space, generating a bistability where a stable steady state coexists with a massive oscillatory attractor. From a dynamical systems perspective, this transition points to a rich topological complexity likely involving global bifurcations—potentially homoclinic or Shilnikov-like—associated with saddle-focus equilibria. Characterizing the boundaries of these global transitions remains a compelling open problem. Our mathematically transparent $\sigma_{21}=0$ baseline provides an ideal scaffold for mapping these complex global dynamics and extending methods used to chart multi-strain structures~\cite{Aguiar_Mapping}.

Our numerical investigations (Sections~\ref{sec:extensions_numerical} and~\ref{sec:general_two_strain}) strongly suggest that the fundamental instability mechanisms driving limit-cycle oscillations and high-amplitude outbreaks are structurally robust. Oscillations persist when incorporating temporary cross-immunity (waning immunity), quarantine measures, or secondary infections from both strains. While our exploration is not exhaustive, it leaves the characterization of their complete dynamical landscape for future work.

The framework developed here may provide a theoretical basis for further mathematical investigation of multi-strain dynamics in epidemiological systems. The core assumption of precluding secondary infections for one strain ($\sigma_{21}=0$) maps naturally to the strongly asymmetric cross-immunity observed between \textit{Bordetella pertussis} and \textit{Bordetella parapertussis}~\cite{wolfe2007antigen}. Permitting secondary infections for both strains ($\sigma_{12}, \sigma_{21} > 0$), as shown in Figure~\ref{fig:bifurcation_diagram_base_nonzero_sigma21}B, provides a baseline for studying antigenically drifting respiratory pathogens like influenza, where bidirectional partial cross-immunity shapes epidemic cycles~\cite{andreasen1997dynamics}. The framework also accommodates targeted extensions for ADE and temporary waning immunity (Figures~\ref{fig:bifurcationDiagram_gamma12_02_extended} and~\ref{fig:bifurcation_diagram_base_nonzero_sigma21}C). While both features are classically linked to the macroscopic cycles of dengue fever~\cite{ferguson1999effect, aguiar2022mathematical}, our extensions treat them in isolation to rigorously understand their individual impacts. Directly applying this framework to dengue requires incorporating additional immunological complexities, such as the temporary cross-immunity modeled in recent literature~\cite{aguiar2008epidemiology,aguiar2013how,aguiar2022mathematical}. Applying these findings to specific, multifaceted epidemiological systems warrants further study, potentially by integrating these isolated mechanisms into composite models.

In conclusion, this work demonstrates that targeted structural simplifications, specifically precluding secondary infections for a single strain ($\sigma_{21}=0$), can successfully resolve severe algebraic barriers without sacrificing biological relevance~\cite{wolfe2007antigen}. By mapping the robust coexistence of local limit-cycle oscillations and macroscopic outbreak cycles, we provide a mathematically transparent framework to study complex multi-strain systems. This baseline not only refines restrictive historical limitations but also establishes a reliable scaffold for analyzing global transitions and higher-dimensional extensions. Ultimately, this analytically tractable approach yields insights into the mechanisms driving multi-strain epidemic cycles, serving as a foundation for future mathematical and epidemiological investigations.

\backmatter
\bmhead{Data \& Code Availability Statement}
We have used Matlab version 2025b for the computations in this work.  The Matlab codes used to perform the computations and produce the graphs presented in the manuscript are openly available at {\tt https://github.com/NGavish/AnalyticallyTractableMultiStrain.git}

No data was analyzed for this manuscript.

\bmhead{Statements and Declarations}

{\bf Competing Interests} The authors have no competing interests to declare.

{\bf Funding} This research was supported by the Israel Science Foundation (grant no. 3730/20) within the KillCorona-Curbing Coronavirus Research Program, and by the Israeli Science Foundation (ISF) grant 1596/23.
\bmhead{Acknowledgements}
The author thanks Arik Yochelis and Guy Katriel for their valuable comments and discussions on this research.

\begin{appendices}

\section{Proof of Proposition~\ref{prop:stability_EE_extended}}\label{app:proof_phiEE_extended}
We note that this proof below encompasses the proof of Proposition~\ref{prop:single_strain_stability} as a special case recovered by setting~$w=w_1=w_2=\delta_1=\delta_2=0$.

We first consider the stability of~$\phi^{EE,1}$.  The Jacobian of~\eqref{eq:extended_model} at~$\phi^{EE,1}$ has the corresponding characteristic polynomial
\begin{subequations}\label{eq:characteristicPolynomialEE1}
\begin{equation}\label{eq:PEE1}
P^{EE,1}(\lambda)=(\lambda+\alpha_2 + \mu)(\lambda + w_2+\mu)(\lambda+w+\mu)(\lambda + \mu)Q_3^{EE,1}(\lambda)Q_2^{EE,1}(\lambda)
\end{equation}
where
\begin{equation}
\begin{split}
Q_3^{EE,1}(\lambda)&=\lambda^3+b_3^{EE,1}\lambda^2+c_3^{EE,1}\lambda+d_3^{EE,1},\\
b_3^{EE,1}&=2\mu + \alpha_1 +(\mu + \gamma_1 + \delta_1)\mathcal{R}_1I_1^{EE,1},\\
c_3^{EE,1}&=(\mu + \alpha_1)(\mu + w_1) + (2\mu + \delta_1 + \gamma_1 + w_1 + \alpha_1)(\mu + \gamma_1 + \delta_1)\mathcal{R}_1I_1^{EE,1},\\
d_3^{EE,1}&=(\mu^2 + (w_1 + \alpha_1 + \delta_1 + \gamma_1)\mu + (w_1 + \delta_1 + \gamma_1)\alpha_1 + \delta_1w_1)(\mu + \gamma_1 + \delta_1)\mathcal{R}_1I_1^{EE,1},\\
\end{split}
\end{equation}
and
\begin{equation}
\mathcal{R}_1Q_2^{EE,1}(\lambda)=\lambda^2+b_2^{EE,1} \lambda+(\gamma_{12}+\mu)(\gamma_2+\delta_2+\mu)(1-\hat{\mathcal{R}}_2^1),
\end{equation}
where
\[
b_2^{EE,1}=\mu+(\gamma_{12} + \mu)(1-\hat{\mathcal{R}}_2^1) + \gamma_{12}\frac{\mathcal{R}_2}{\mathcal{R}_1} + (\delta_2 + \gamma_2)\frac{(\gamma_{12}+\mu)(1-\hat{\mathcal{R}}_2^1)+\eta_2\sigma_{12}(\mu + \gamma_2 + \delta_2)\mathcal{R}_1R_1^{EE,1}}{\mu+\gamma_{12}+\eta_2\sigma_{12}(\mu + \gamma_2 + \delta_2)\mathcal{R}_1R_1^{EE,1}}.
\]
\end{subequations}

The characteristic polynomial~,$P^{EE,1}$, has nine roots.   
Four are the negative roots
\[
\lambda_1=-\mu-w_2<0,\quad \lambda_2=-\mu-\alpha_2<0,\quad \lambda_3=-\mu-w<0,\quad \lambda_4=-\mu<0,
\]
three are the roots of~$Q_3^{EE,1}(\lambda)$ and two additional roots that are the roots of~$Q_2^{EE,1}(\lambda)$.
All coefficients of~$Q_3^{EE,1}(\lambda)$ are strictly positive and satisfy~$b^{EE,1}c_3^{EE,1}-d_3^{EE,1}>0$. 
Thus, according to the Routh-Hurwitz criterion, all roots of~$Q_3^{EE,1}$ have a negative real part. When~$\hat{\mathcal{R}}_2^1<1$, all coefficients of~$Q_3^{EE,1}(\lambda)$ are strictly positive and thus, according to the Routh-Hurwitz criterion, all roots of~$Q_2^{EE,1}$ have a negative real part.  Otherwise,  
if~$\hat{\mathcal{R}}_2^1\ge1$, then~$Q_2^{EE,1}(\lambda)$ has a non-negative root.

Similarly, the Jacobian of~\eqref{eq:extended_model} at~$\phi^{EE,2}$ has the corresponding characteristic polynomial
\begin{subequations}\label{eq:characteristicPolynomialEE2}
\begin{equation}
\begin{split}
P^{EE,2}(\lambda)&=\frac1{\mathcal{R}_2}(\lambda + \alpha_1 + \mu)(\lambda + \mu + w)(\lambda + \gamma_{12} + \mu)(\lambda+\mu)\times\\&\left[\lambda + \mu + w_1 + (\delta_2+\gamma_2+\mu)\sigma_{12}\mathcal{R}_2I_2^{EE,2}\eta_2 \right]\times\\&\left[\lambda+(\mu + \gamma_1 + \delta_1)(1-\hat{\mathcal{R}}_1^2)\right]Q_3^{EE,2}(\lambda),
\end{split}
    \end{equation}
where
\begin{equation}
\begin{split}
Q_3^{EE,2}(\lambda)&=\lambda^3+b_3^{EE,2}\lambda^2+c_3^{EE,2}\lambda+d_3^{EE,2},\\
b_3^{EE,2}&=(\mu + \gamma_2 + \delta_2)\mathcal{R}_2I_2^{EE,2} + 2\mu + w_2 + \alpha_2\\
c_3^{EE,2}&=(2\mu + w_2 + \alpha_2 + \delta_2 + \gamma_2)(\mu + \gamma_2 + \delta_2)\mathcal{R}_2I_2^{EE,2} + (\mu + \alpha_2)(\mu + w_2)\\
d_3^{EE,2}&=(\mu^2 + (w_2 + \alpha_2 + \delta_2 + \gamma_2)\mu + (w_2 + \alpha_2)\delta_2 + \alpha_2(\gamma_2 + w_2))(\mu + \gamma_2 + \delta_2)\mathcal{R}_2I_2^{EE,2}.
\end{split}
\end{equation}
\end{subequations}
The characteristic polynomial,~$P^{EE,2}(\lambda)$, has nine roots. Three of them are roots of~$Q_3^{EE,2}(\lambda)$. All coefficients of~$Q_3^{EE,2}(\lambda)$ are strictly positive, and~$b_3^{EE,2}c_3^{EE,2}-a_3^{EE,2}>0$. Thus, by the Routh-Hurwitz criterion, the three roots of~$Q_3^{EE,2}(\lambda)$ have a negative real part. An additional root of~$P^{EE,2}(\lambda)$ is
\[
\lambda=-\frac{(1-\hat{\mathcal{R}}_1^2)(\mu + \gamma_1 + \delta_1)}{\mathcal{R}_2},
\]
which is negative if and only if~$\hat{\mathcal{R}}_1^2<1$. The remaining roots of~$P^{EE,2}(\lambda)$ are strictly negative due to~\eqref{eq:parameterspace_reduced_only} and~\eqref{eq:parameterspace_extended}.

\section{Proof of Lemma~\ref{lem:curveGamma}}\label{app:proof_lem_curveGamma}
Let~$\mathcal{R}_1^*>1$. Then, the range of values of~$\mathcal{R}_2^*$ for which~$(\mathcal{R}_1^*,\mathcal{R}_2^*)\in\Omegamu$ is
\begin{equation}\label{eq:R2*_bounds}
\mathcal{R}_2^{\rm bound}(\mathcal{R}_1^*)<R_2^*<\mathcal{R}_1^*,
\end{equation}
where~$\mathcal{R}_2^{\rm bound}(\mathcal{R}_1^*)$ satisfies~$\hat{\mathcal{R}}_2^1(\mathcal{R}_1^*,\mathcal{R}_2^{\rm bound})=1$ and equals
\[
\mathcal{R}_2^{\rm bound}(\mathcal{R}_1^*)=\frac{\mathcal{R}_1^*}{1+\frac{\gamma_1(\gamma_2+\mu)\eta_2\sigma_{12}}{(\gamma_1+\mu)(\gamma_{12}+\mu)}(\mathcal{R}_1^*-1)}.
\]
Let~$h(\mathcal{R}_2;\mathcal{R}_1,\gamma_1,\gamma_2,\gamma_{12},\sigma_{12},\eta_2)=A-B$. 
Then, at the left bound of the interval~\eqref{eq:R2*_bounds},
\[
\gamma_{12}h\left(\mathcal{R}_2=\mathcal{R}_2^{\rm bound}\right)=-\gamma_{12}B\left(\mathcal{R}_2=\mathcal{R}_2^{\rm bound}\right)<0,
\]
while at the right bound of~\eqref{eq:R2*_bounds},
\[
(\gamma_1+\mu)(\gamma_{12}+\mu)h(\mathcal{R}_2=\mathcal{R}_1)=\eta_2\sigma_{12}\gamma_2\gamma_{12}\gamma_1(\gamma_2+\mu)(\mathcal{R}_1 - 1)\mathcal{R}_1^2>0.
\]
Thus, for every~$\mathcal{R}_1^*>1$, there exists a root~$\mathcal{R}_2^*$ of~$h$ that satisfies~\eqref{eq:R2*_bounds}. Since~$h$ is a quadratic function in~$\mathcal{R}_2$, it follows that the additional root of~$h$ does not satisfy~\eqref{eq:R2*_bounds}. Therefore,~$h=A-B=0$ only along the curve
\[
\Gamma^*=\{(\mathcal{R}_1^*,\mathcal{R}_2^*(\mathcal{R}_1^*)) \mid \mathcal{R}_1^*>1\}\subset\Omegamu.
\]
\section{Numerical details}\label{app:numericalDetails}
All numerical simulations throughout this work are implemented using \texttt{Matlab} \texttt{ode45} routine, configured with a maximum step size of $0.05$ and an absolute tolerance of $10^{-14}$.
\section{Additional numerical examples}\label{app:numerical_examples}
\begin{figure}[ht!]
\begin{center}
\includegraphics[width=\textwidth]{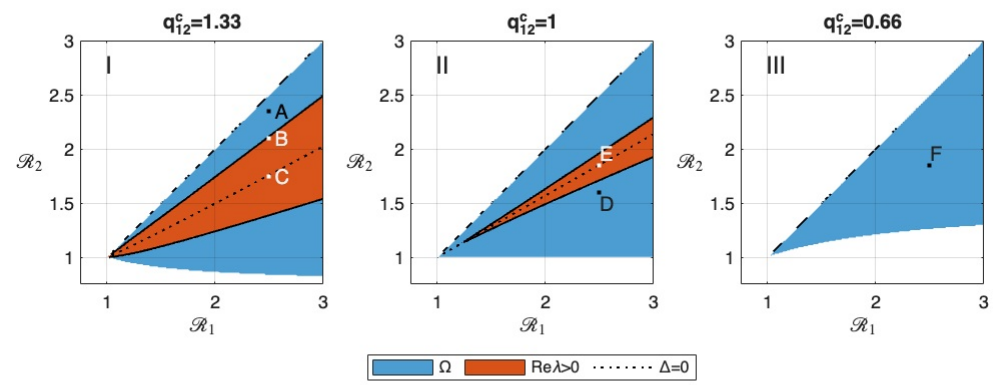}
\caption{Bifurcation diagram for the base model~\eqref{eq:base_model} as a function of~$\mathcal{R}_1$ and~$\mathcal{R}_2$. The parameters~$\beta_1$ and~$\beta_2$ vary while the other parameters are kept fixed: $\gamma_1=1$,~$\gamma_2=0.75$,~$\sigma_{12}=0.8$, $\eta_2=1$, $\mu=2.5\cdot10^{-4}$ and  
I: ~$q_{12}^c=2$.  II:~$q_{12}^c=1$. III:~$q_{12}^c=2/3$. The diagram shows the region~$\Omegamu$ where the coexistence steady-state exists (blue) and the region where it is unstable (red). The curve~$\Gamma^*$ is shown as a dotted black curve. Points A-F are at~$\mathcal{R}_1=2.5$ with the following $\mathcal{R}_2$ values: A(2.35), B(2.1), C(1.75), D(1.6), E(1.85), and F(1.85).}
\label{fig:bifurcation_diagram_base}
\end{center}
\end{figure}Figure~\ref{fig:bifurcation_diagram_base} presents a bifurcation diagram for the base model~\eqref{eq:base_model} with the parameters~$\gamma_1=1$,~$\gamma_2=0.75$,~$\sigma_{12}=0.8$ and~$\eta_2=1$, while Figure~\ref{fig:SolutionsSample} presents solutions~$I_1(t)$ and~$I_2(t)$ of~\eqref{eq:base_model} corresponding to points A-F in Figure~\ref{fig:bifurcation_diagram_base}.  As expected, the results are qualitatively the same of the bifurcation diagram presented in Figures~\ref{fig:bifurcation_diagram_base_alternative} and~\ref{fig:SolutionsSample_alternative} for the case~$\gamma_2=1.1$,~$\sigma_{12}=1.1$ and~$\eta_2=1.1$.   This enforces the conclusion arising from the analysis that the key parameter determining the stability behavior of~$\phi^{CE}$ is~$q_{12}^c$, rather than the parameters that it depends on.  

\begin{figure}[ht!]
\begin{center}
\includegraphics[width=\textwidth]{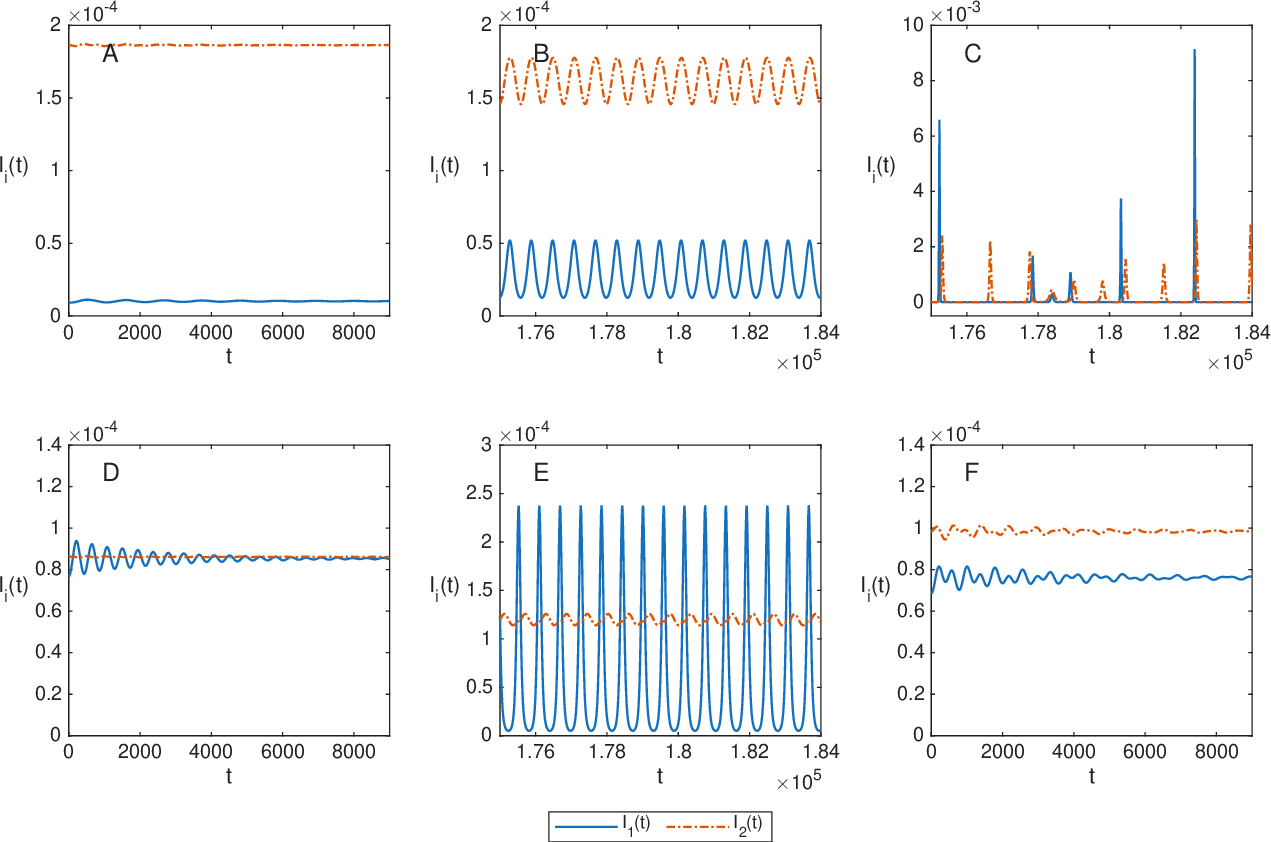}
\caption{Solutions~$I_1(t)$ and~$I_2(t)$ of~\eqref{eq:base_model} corresponding to points A-F in Figure~\ref{fig:bifurcation_diagram_base}.}
\label{fig:SolutionsSample}
\end{center}
\end{figure}
\section{Proof of Proposition~\ref{prop:phiCE_extended}}\label{app:proof_phiCE_extended}
Similar to the proof of Proposition~\ref{prop:phiCE}, from~\eqref{eq:full_eqI1},~$I_1\ne0$ implies
\[
S^{CE}=\frac1{\mathcal{R}_1}.
\]
From~\eqref{eq:full_eqI2},
\[
I_{12}^{CE}=cI_2^{CE},\quad c=\frac{\mathcal{R}_1 - \mathcal{R}_2}{\eta_2\mathcal{R}_2}.
\]
Similarly, from~\eqref{eq:full_eqQ2},~\eqref{eq:full_eqI12},~\eqref{eq:full_eqR2}, and~\eqref{eq:full_eqR},
\[
\begin{split}
Q_2^{CE}&=\frac{\delta_2}{\mu+\alpha_2}I_2^{CE},\quad R_2^{CE}=\frac{(\delta_2 + \gamma_2)\alpha_2 + \mu\gamma_2}{(\mu + \alpha_2)(\mu + w_2)}I_2^{CE},\\ R^{CE}&=\frac{\gamma_{12}c}{\mu+w}I_2^{CE},
\quad R_1^{CE} = \frac{(\mu + \gamma_{12})c}{(\mu + \gamma_2 + \delta_2)\mathcal{R}_2\sigma_{12}(1+c\eta_2)}.
\end{split}
\]
Thus, from~\eqref{eq:full_eqQ1} and~\eqref{eq:full_eqR1},
\[
\left(\gamma_1+\alpha_1\frac{\delta_1}{\mu+\alpha_1}\right)I_1^{CE}=(\mu+\gamma_{12})R_1^{CE}cI_2^{CE}+(\mu+w_1)R_1^{CE},
\]
or
\[
I_1^{CE}=aI_2^{CE}+b,\quad a=\frac{(\mu+\alpha_1)(\mu+\gamma_{12})c}{(\mu+\alpha_1)\gamma_1+\alpha_1\delta_1}R_1^{CE},\quad b=\frac{(\mu+\alpha_1)(\mu+w_1)}{(\mu+\alpha_1)\gamma_1+\alpha_1\delta_1}R_1^{CE}.
\]
Summing~\eqref{eq:full_eqS},~\eqref{eq:full_eqI1}, and~\eqref{eq:full_eqI2}, and isolating~$I_2^{CE}$ yields
\[
I_2^{CE} = \frac{(1-S-b)\mu - b\gamma_1}{(a + 1)\mu + a\gamma_1 + \gamma_2}.
\]
Direct computation shows that~$I_1^{CE}>0$ and~$I_2^{CE}>0$ only in the parameter regime~\eqref{eq:omega_mu}.
\end{appendices}



\begin{thebibliography}{28}
\ifx \bisbn   \undefined \def \bisbn  #1{ISBN #1}\fi
\ifx \binits  \undefined \def \binits#1{#1}\fi
\ifx \bauthor  \undefined \def \bauthor#1{#1}\fi
\ifx \batitle  \undefined \def \batitle#1{#1}\fi
\ifx \bjtitle  \undefined \def \bjtitle#1{#1}\fi
\ifx \bvolume  \undefined \def \bvolume#1{\textbf{#1}}\fi
\ifx \byear  \undefined \def \byear#1{#1}\fi
\ifx \bissue  \undefined \def \bissue#1{#1}\fi
\ifx \bfpage  \undefined \def \bfpage#1{#1}\fi
\ifx \blpage  \undefined \def \blpage #1{#1}\fi
\ifx \burl  \undefined \def \burl#1{\textsf{#1}}\fi
\ifx \doiurl  \undefined \def \doiurl#1{\url{https://doi.org/#1}}\fi
\ifx \betal  \undefined \def \betal{\textit{et al.}}\fi
\ifx \binstitute  \undefined \def \binstitute#1{#1}\fi
\ifx \binstitutionaled  \undefined \def \binstitutionaled#1{#1}\fi
\ifx \bctitle  \undefined \def \bctitle#1{#1}\fi
\ifx \beditor  \undefined \def \beditor#1{#1}\fi
\ifx \bpublisher  \undefined \def \bpublisher#1{#1}\fi
\ifx \bbtitle  \undefined \def \bbtitle#1{#1}\fi
\ifx \bedition  \undefined \def \bedition#1{#1}\fi
\ifx \bseriesno  \undefined \def \bseriesno#1{#1}\fi
\ifx \blocation  \undefined \def \blocation#1{#1}\fi
\ifx \bsertitle  \undefined \def \bsertitle#1{#1}\fi
\ifx \bsnm \undefined \def \bsnm#1{#1}\fi
\ifx \bsuffix \undefined \def \bsuffix#1{#1}\fi
\ifx \bparticle \undefined \def \bparticle#1{#1}\fi
\ifx \barticle \undefined \def \barticle#1{#1}\fi
\bibcommenthead
\ifx \bconfdate \undefined \def \bconfdate #1{#1}\fi
\ifx \botherref \undefined \def \botherref #1{#1}\fi
\ifx \url \undefined \def \url#1{\textsf{#1}}\fi
\ifx \bchapter \undefined \def \bchapter#1{#1}\fi
\ifx \bbook \undefined \def \bbook#1{#1}\fi
\ifx \bcomment \undefined \def \bcomment#1{#1}\fi
\ifx \oauthor \undefined \def \oauthor#1{#1}\fi
\ifx \citeauthoryear \undefined \def \citeauthoryear#1{#1}\fi
\ifx \endbibitem  \undefined \def \endbibitem {}\fi
\ifx \bconflocation  \undefined \def \bconflocation#1{#1}\fi
\ifx \arxivurl  \undefined \def \arxivurl#1{\textsf{#1}}\fi
\csname PreBibitemsHook\endcsname

%%% 1
\bibitem[\protect\citeauthoryear{Martcheva}{2015}]{martcheva2015introduction}
\begin{bbook}
\bauthor{\bsnm{Martcheva}, \binits{M.}}:
\bbtitle{An Introduction to Mathematical Epidemiology}
vol. \bseriesno{61}.
\bpublisher{Springer},
\blocation{New York}
(\byear{2015})
\end{bbook}
\endbibitem

%%% 2
\bibitem[\protect\citeauthoryear{Wormser and
  Pourbohloul}{2008}]{wormser2008modeling}
\begin{botherref}
\oauthor{\bsnm{Wormser}, \binits{G.P.}},
\oauthor{\bsnm{Pourbohloul}, \binits{B.}}:
Modeling Infectious Diseases in Humans and Animals By Matthew James Keeling and
  Pejman Rohani Princeton, NJ: Princeton University Press, 2008. 408 pp.,
  Illustrated (hardcover).
The University of Chicago Press
(2008)
\end{botherref}
\endbibitem

%%% 3
\bibitem[\protect\citeauthoryear{Aguiar et~al.}{2022}]{aguiar2022mathematical}
\begin{barticle}
\bauthor{\bsnm{Aguiar}, \binits{M.}},
\bauthor{\bsnm{Anam}, \binits{V.}},
\bauthor{\bsnm{Blyuss}, \binits{K.B.}},
\bauthor{\bsnm{Estadilla}, \binits{C.D.S.}},
\bauthor{\bsnm{Guerrero}, \binits{B.V.}},
\bauthor{\bsnm{Knopoff}, \binits{D.}},
\bauthor{\bsnm{Kooi}, \binits{B.W.}},
\bauthor{\bsnm{Srivastav}, \binits{A.K.}},
\bauthor{\bsnm{Steindorf}, \binits{V.}},
\bauthor{\bsnm{Stollenwerk}, \binits{N.}}:
\batitle{Mathematical models for dengue fever epidemiology: A 10-year
  systematic review}.
\bjtitle{Physics of Life Reviews}
\bvolume{40},
\bfpage{65}--\blpage{92}
(\byear{2022})
\end{barticle}
\endbibitem

%%% 4
\bibitem[\protect\citeauthoryear{Keeling and Rohani}{2008}]{Keeling-book}
\begin{bbook}
\bauthor{\bsnm{Keeling}, \binits{M.J.}},
\bauthor{\bsnm{Rohani}, \binits{P.}}:
\bbtitle{Modeling Infectious Diseases in Humans and Animals}.
\bpublisher{Princeton University Press},
\blocation{Princeton}
(\byear{2008})
\end{bbook}
\endbibitem

%%% 5
\bibitem[\protect\citeauthoryear{Gumel and Song}{2008}]{gumel2008existence}
\begin{barticle}
\bauthor{\bsnm{Gumel}, \binits{A.B.}},
\bauthor{\bsnm{Song}, \binits{B.}}:
\batitle{Existence of multiple-stable equilibria for a multi-drug-resistant
  model of mycobacterium tuberculosis}.
\bjtitle{Mathematical Biosciences \& Engineering}
\bvolume{5}(\bissue{3}),
\bfpage{437}--\blpage{455}
(\byear{2008})
\end{barticle}
\endbibitem

%%% 6
\bibitem[\protect\citeauthoryear{Garba et~al.}{2013}]{garba2013cross}
\begin{barticle}
\bauthor{\bsnm{Garba}, \binits{S.}},
\bauthor{\bsnm{Safi}, \binits{M.A.}},
\bauthor{\bsnm{Gumel}, \binits{A.}}:
\batitle{Cross-immunity-induced backward bifurcation for a model of
  transmission dynamics of two strains of influenza}.
\bjtitle{Nonlinear Analysis: Real World Applications}
\bvolume{14}(\bissue{3}),
\bfpage{1384}--\blpage{1403}
(\byear{2013})
\end{barticle}
\endbibitem

%%% 7
\bibitem[\protect\citeauthoryear{Hussaini
  et~al.}{2016}]{hussaini2016mathematical}
\begin{barticle}
\bauthor{\bsnm{Hussaini}, \binits{N.}},
\bauthor{\bsnm{Lubuma}, \binits{J.M.}},
\bauthor{\bsnm{Barley}, \binits{K.}},
\bauthor{\bsnm{Gumel}, \binits{A.}}:
\batitle{Mathematical analysis of a model for avl--hiv co-endemicity}.
\bjtitle{Mathematical biosciences}
\bvolume{271},
\bfpage{80}--\blpage{95}
(\byear{2016})
\end{barticle}
\endbibitem

%%% 8
\bibitem[\protect\citeauthoryear{Ferguson et~al.}{1999}]{ferguson1999effect}
\begin{barticle}
\bauthor{\bsnm{Ferguson}, \binits{N.}},
\bauthor{\bsnm{Anderson}, \binits{R.}},
\bauthor{\bsnm{Gupta}, \binits{S.}}:
\batitle{The effect of antibody-dependent enhancement on the transmission
  dynamics and persistence of multiple-strain pathogens}.
\bjtitle{Proceedings of the National Academy of Sciences}
\bvolume{96}(\bissue{2}),
\bfpage{790}--\blpage{794}
(\byear{1999})
\end{barticle}
\endbibitem

%%% 9
\bibitem[\protect\citeauthoryear{Gupta et~al.}{1998}]{gupta1998chaos}
\begin{barticle}
\bauthor{\bsnm{Gupta}, \binits{S.}},
\bauthor{\bsnm{Ferguson}, \binits{N.}},
\bauthor{\bsnm{Anderson}, \binits{R.}}:
\batitle{Chaos, persistence, and evolution of strain structure in antigenically
  diverse infectious agents}.
\bjtitle{Science}
\bvolume{280}(\bissue{5365}),
\bfpage{912}--\blpage{915}
(\byear{1998})
\end{barticle}
\endbibitem

%%% 10
\bibitem[\protect\citeauthoryear{Nu{\~n}o et~al.}{2005}]{nuno2005dynamics}
\begin{barticle}
\bauthor{\bsnm{Nu{\~n}o}, \binits{M.}},
\bauthor{\bsnm{Feng}, \binits{Z.}},
\bauthor{\bsnm{Martcheva}, \binits{M.}},
\bauthor{\bsnm{Castillo-Chavez}, \binits{C.}}:
\batitle{Dynamics of two-strain influenza with isolation and partial
  cross-immunity}.
\bjtitle{SIAM Journal on Applied Mathematics}
\bvolume{65}(\bissue{3}),
\bfpage{964}--\blpage{982}
(\byear{2005})
\end{barticle}
\endbibitem

%%% 11
\bibitem[\protect\citeauthoryear{Nuno et~al.}{2008}]{nuno2008mathematical}
\begin{botherref}
\oauthor{\bsnm{Nuno}, \binits{M.}},
\oauthor{\bsnm{Castillo-Chavez}, \binits{C.}},
\oauthor{\bsnm{Feng}, \binits{Z.}},
\oauthor{\bsnm{Martcheva}, \binits{M.}}:
Mathematical models of influenza: the role of cross-immunity, quarantine and
  age-structure.
Mathematical Epidemiology,
349--364
(2008)
\end{botherref}
\endbibitem

%%% 12
\bibitem[\protect\citeauthoryear{Thieme}{2007}]{thieme2007pathogen}
\begin{botherref}
\oauthor{\bsnm{Thieme}, \binits{H.R.}}:
Pathogen competition and coexistence and the evolution of virulence.
Mathematics for Life Science and Medicine,
123--153
(2007)
\end{botherref}
\endbibitem

%%% 13
\bibitem[\protect\citeauthoryear{Kuddus et~al.}{2022}]{kuddus2022analysis}
\begin{barticle}
\bauthor{\bsnm{Kuddus}, \binits{M.A.}},
\bauthor{\bsnm{McBryde}, \binits{E.S.}},
\bauthor{\bsnm{Adekunle}, \binits{A.I.}},
\bauthor{\bsnm{Meehan}, \binits{M.T.}}:
\batitle{Analysis and simulation of a two-strain disease model with nonlinear
  incidence}.
\bjtitle{Chaos, Solitons \& Fractals}
\bvolume{155},
\bfpage{111637}
(\byear{2022})
\end{barticle}
\endbibitem

%%% 14
\bibitem[\protect\citeauthoryear{Castillo-Chavez
  et~al.}{1989}]{castillo1989epidemiological}
\begin{barticle}
\bauthor{\bsnm{Castillo-Chavez}, \binits{C.}},
\bauthor{\bsnm{Hethcote}, \binits{H.W.}},
\bauthor{\bsnm{Andreasen}, \binits{V.}},
\bauthor{\bsnm{Levin}, \binits{S.A.}},
\bauthor{\bsnm{Liu}, \binits{W.M.}}:
\batitle{Epidemiological models with age structure, proportionate mixing, and
  cross-immunity}.
\bjtitle{Journal of Mathematical Biology}
\bvolume{27},
\bfpage{233}--\blpage{258}
(\byear{1989})
\end{barticle}
\endbibitem

%%% 15
\bibitem[\protect\citeauthoryear{Kooi et~al.}{2014}]{kooi2014analysis}
\begin{barticle}
\bauthor{\bsnm{Kooi}, \binits{B.W.}},
\bauthor{\bsnm{Aguiar}, \binits{M.}},
\bauthor{\bsnm{Stollenwerk}, \binits{N.}}:
\batitle{Analysis of an asymmetric two-strain dengue model}.
\bjtitle{Mathematical biosciences}
\bvolume{248},
\bfpage{128}--\blpage{139}
(\byear{2014})
\end{barticle}
\endbibitem

%%% 16
\bibitem[\protect\citeauthoryear{Kooi et~al.}{2023}]{kooi2023multi}
\begin{bchapter}
\bauthor{\bsnm{Kooi}, \binits{B.W.}},
\bauthor{\bsnm{Rashkov}, \binits{P.}},
\bauthor{\bsnm{Venturino}, \binits{E.}}:
\bctitle{Multi-strain host-vector dengue modeling: Dynamics and control}.
In: \beditor{\bsnm{Ghaffari}, \binits{P.}} (ed.)
\bbtitle{Bio-mathematics, Statistics, and Nano-Technologies: Mosquito Control
  Strategies},
pp. \bfpage{110}--\blpage{142}.
\bpublisher{Chapman and Hall/CRC},
\blocation{Boca Raton}
(\byear{2023})
\end{bchapter}
\endbibitem

%%% 17
\bibitem[\protect\citeauthoryear{Kooi et~al.}{2013}]{kooi2013bifurcation}
\begin{barticle}
\bauthor{\bsnm{Kooi}, \binits{B.W.}},
\bauthor{\bsnm{Aguiar}, \binits{M.}},
\bauthor{\bsnm{Stollenwerk}, \binits{N.}}:
\batitle{Bifurcation analysis of a family of multi-strain epidemiology models}.
\bjtitle{Journal of computational and applied mathematics}
\bvolume{252},
\bfpage{148}--\blpage{158}
(\byear{2013})
\end{barticle}
\endbibitem

%%% 18
\bibitem[\protect\citeauthoryear{Aguiar et~al.}{2011}]{aguiar2011role}
\begin{barticle}
\bauthor{\bsnm{Aguiar}, \binits{M.}},
\bauthor{\bsnm{Ballesteros}, \binits{S.}},
\bauthor{\bsnm{Kooi}, \binits{B.W.}},
\bauthor{\bsnm{Stollenwerk}, \binits{N.}}:
\batitle{The role of seasonality and import in a minimalistic multi-strain
  dengue model capturing differences between primary and secondary infections:
  complex dynamics and its implications for data analysis}.
\bjtitle{Journal of theoretical biology}
\bvolume{289},
\bfpage{181}--\blpage{196}
(\byear{2011})
\end{barticle}
\endbibitem

%%% 19
\bibitem[\protect\citeauthoryear{Aguiar et~al.}{2013}]{aguiar2013how}
\begin{barticle}
\bauthor{\bsnm{Aguiar}, \binits{M.}},
\bauthor{\bsnm{Kooi}, \binits{B.W.}},
\bauthor{\bsnm{Rocha}, \binits{F.}},
\bauthor{\bsnm{Ghaffari}, \binits{P.}},
\bauthor{\bsnm{Stollenwerk}, \binits{N.}}:
\batitle{How much complexity is needed to describe the fluctuations observed in
  dengue haemorrhagic fever incidence data?}
\bjtitle{Ecological Complexity}
\bvolume{16},
\bfpage{31}--\blpage{40}
(\byear{2013})
\end{barticle}
\endbibitem

%%% 20
\bibitem[\protect\citeauthoryear{Aguiar et~al.}{2008}]{aguiar2008epidemiology}
\begin{barticle}
\bauthor{\bsnm{Aguiar}, \binits{M.}},
\bauthor{\bsnm{Kooi}, \binits{B.}},
\bauthor{\bsnm{Stollenwerk}, \binits{N.}}:
\batitle{Epidemiology of dengue fever: A model with temporary cross-immunityand
  possible secondary infection shows bifurcationsand chaotic behaviour in wide
  parameter regions}.
\bjtitle{Mathematical Modelling of Natural Phenomena}
\bvolume{3}(\bissue{4}),
\bfpage{48}--\blpage{70}
(\byear{2008})
\end{barticle}
\endbibitem

%%% 21
\bibitem[\protect\citeauthoryear{Andreasen
  et~al.}{1997}]{andreasen1997dynamics}
\begin{barticle}
\bauthor{\bsnm{Andreasen}, \binits{V.}},
\bauthor{\bsnm{Lin}, \binits{J.}},
\bauthor{\bsnm{Levin}, \binits{S.A.}}:
\batitle{The dynamics of cocirculating influenza strains conferring partial
  cross-immunity}.
\bjtitle{Journal of mathematical biology}
\bvolume{35},
\bfpage{825}--\blpage{842}
(\byear{1997})
\end{barticle}
\endbibitem

%%% 22
\bibitem[\protect\citeauthoryear{Chung and Lui}{2016}]{chung2016dynamics}
\begin{barticle}
\bauthor{\bsnm{Chung}, \binits{K.}},
\bauthor{\bsnm{Lui}, \binits{R.}}:
\batitle{Dynamics of two-strain influenza model with cross-immunity and no
  quarantine class}.
\bjtitle{Journal of Mathematical Biology}
\bvolume{73},
\bfpage{1467}--\blpage{1489}
(\byear{2016})
\end{barticle}
\endbibitem

%%% 23
\bibitem[\protect\citeauthoryear{Gavish}{2026}]{gavish2024newoscillatoryregimetwostrain}
\begin{botherref}
\oauthor{\bsnm{Gavish}, \binits{N.}}:
Hidden Oscillations: How Minimal Immunological Interactions Destabilize
  Two-Strain Epidemic Models.
Accepted for publication in SIAM Journal on Life Sciences (SIALS)
(2026)
\end{botherref}
\endbibitem

%%% 24
\bibitem[\protect\citeauthoryear{Wolfe et~al.}{2007}]{wolfe2007antigen}
\begin{barticle}
\bauthor{\bsnm{Wolfe}, \binits{D.N.}},
\bauthor{\bsnm{Goebel}, \binits{E.M.}},
\bauthor{\bsnm{Bjornstad}, \binits{O.N.}},
\bauthor{\bsnm{Restif}, \binits{O.}},
\bauthor{\bsnm{Harvill}, \binits{E.T.}}:
\batitle{The o antigen enables bordetella parapertussis to avoid bordetella
  pertussis-induced immunity}.
\bjtitle{Infection and immunity}
\bvolume{75}(\bissue{10}),
\bfpage{4972}--\blpage{4979}
(\byear{2007})
\end{barticle}
\endbibitem

%%% 25
\bibitem[\protect\citeauthoryear{Dawes and Gog}{2002}]{dawes2002onset}
\begin{barticle}
\bauthor{\bsnm{Dawes}, \binits{J.}},
\bauthor{\bsnm{Gog}, \binits{J.}}:
\batitle{The onset of oscillatory dynamics in models of multiple disease
  strains}.
\bjtitle{Journal of Mathematical Biology}
\bvolume{45}(\bissue{6}),
\bfpage{471}--\blpage{510}
(\byear{2002})
\end{barticle}
\endbibitem

%%% 26
\bibitem[\protect\citeauthoryear{Nemhauser et~al.}{2023}]{CDCyellowbook2024}
\begin{botherref}
\oauthor{\bsnm{Nemhauser}, \binits{J.}},
\oauthor{\bsnm{LaRocque}, \binits{R.}},
\oauthor{\bsnm{Alvarado-Ramy}, \binits{F.}},
\oauthor{\bsnm{Angelo}, \binits{K.}},
\oauthor{\bsnm{Ericsson}, \binits{C.}},
\oauthor{\bsnm{Gertz}, \binits{A.}},
\oauthor{\bsnm{Kozarsky}, \binits{P.}},
\oauthor{\bsnm{Ostroff}, \binits{S.}},
\oauthor{\bsnm{Ryan}, \binits{E.}},
\oauthor{\bsnm{Shlim}, \binits{D.}},
\oauthor{\bsnm{Stauffer}, \binits{W.}},
\oauthor{\bsnm{Weinberg}, \binits{M.}},
\oauthor{\bsnm{Wilson}, \binits{M.E.}},
\oauthor{\bsnm{Keir}, \binits{J.}},
\oauthor{\bsnm{Leidel}, \binits{L.}},
\oauthor{\bsnm{Crowe}, \binits{S.}},
\oauthor{\bsnm{Guinn}, \binits{A.}}:
{CDC} {Y}ellow {B}ook 2024: Health information for international travel.
Oxford University Press
(2023)
\end{botherref}
\endbibitem

%%% 27
\bibitem[\protect\citeauthoryear{Gomes and Medley}{2002}]{gomes2002dynamics}
\begin{bchapter}
\bauthor{\bsnm{Gomes}, \binits{M.G.M.}},
\bauthor{\bsnm{Medley}, \binits{G.F.}}:
\bctitle{Dynamics of multiple strains of infectious agents coupled by
  cross-immunity: a comparison of models}.
In: \bbtitle{Mathematical Approaches for Emerging and Reemerging Infectious
  Diseases: Models, Methods, and Theory},
pp. \bfpage{171}--\blpage{191}
(\byear{2002}).
\bcomment{Springer}
\end{bchapter}
\endbibitem

%%% 28
\bibitem[\protect\citeauthoryear{Aguiar et~al.}{2015}]{Aguiar_Mapping}
\begin{barticle}
\bauthor{\bsnm{Aguiar}, \binits{M.}},
\bauthor{\bsnm{Stollenwerk}, \binits{N.}},
\bauthor{\bsnm{Kooi}, \binits{B.W.}}:
\batitle{Charting complex oscillations and global bifurcation structures in
  multi-strain pathogen systems}.
\bjtitle{Journal of Mathematical Biology}
\bvolume{71}(\bissue{3}),
\bfpage{555}--\blpage{581}
(\byear{2015})
\end{barticle}
\endbibitem

\end{thebibliography}
\end{document}